\documentclass[a4paper,twocolumn,11pt,accepted=2022-05-17]{quantumarticle}
\pdfoutput=1
\usepackage[utf8]{inputenc}
\usepackage[english]{babel}
\usepackage[T1]{fontenc}
\usepackage{amsmath}
\usepackage{hyperref}
\usepackage{amsmath,bbm}
\usepackage{amsthm}
\usepackage{latexsym}
\usepackage{amssymb}

\usepackage{tikz}
\usepackage{lipsum}
\usepackage{float}
\usepackage{graphicx}           
\usepackage{color}
\usepackage{comment}
\usepackage{mathtools}

\DeclarePairedDelimiter\floor{\lfloor}{\rfloor}

\newcommand{\be}{\begin{equation}}
	\newcommand{\ee}{\end{equation}}
\newcommand{\bea}{\begin{eqnarray}}
	\newcommand{\eea}{\end{eqnarray}}
\allowdisplaybreaks

\def\squareforqed{\hbox{\rlap{$\sqcap$}$\sqcup$}}
\def\qed{\ifmmode\squareforqed\else{\unskip\nobreak\hfil
		\penalty50\hskip1em\null\nobreak\hfil\squareforqed
		\parfillskip=0pt\finalhyphendemerits=0\endgraf}\fi}
\def\endenv{\ifmmode\;\else{\unskip\nobreak\hfil
		\penalty50\hskip1em\null\nobreak\hfil\;
		\parfillskip=0pt\finalhyphendemerits=0\endgraf}\fi}

\newtheorem{thm}{Theorem}
\newtheorem{lemma}{Lemma}

\newtheorem{assumption}{Assumption}

\begin{document}

\title{Robust certification of arbitrary outcome  quantum measurements from temporal correlations}

\author{Debarshi Das}
\email{dasdebarshi90@gmail.com}
\affiliation{S. N. Bose National Centre for Basic Sciences, Block JD, Sector III, Salt Lake, Kolkata 700106, India }
\affiliation{Department of Physics and Astronomy, University College London, Gower Street, WC1E 6BT London, UK}

\author{Ananda G. Maity}
\email{anandamaity289@gmail.com}
\affiliation{S. N. Bose National Centre for Basic Sciences, Block JD, Sector III, Salt Lake, Kolkata 700106, India }

\author{Debashis Saha}
\email{saha@bose.res.in}
\affiliation{S. N. Bose National Centre for Basic Sciences, Block JD, Sector III, Salt Lake, Kolkata 700106, India }

\author{A. S. Majumdar}
\email{archan@bose.res.in}
\affiliation{S. N. Bose National Centre for Basic Sciences, Block JD, Sector III, Salt Lake, Kolkata 700106, India }

\begin{abstract}
Certification of quantum devices received from  unknown providers is a primary  requirement before utilizing the devices for any information processing  task.  Here, we establish a protocol for certification of  a particular set of $d$-outcome quantum measurements (with $d$ being arbitrary) in a setup comprising of a preparation followed by two measurements in sequence. We propose a set of temporal inequalities pertaining to different $d$ involving correlation functions corresponding to successive measurement outcomes, that are not satisfied by quantum devices. Using quantum violations of these inequalities, we certify  specific $d$-outcome quantum measurements  under some minimal assumptions which can be met in an experiment efficiently. Our certification protocol neither requires entanglement, nor any prior knowledge about the dimension of the system under consideration. We further show that our protocol is robust against  practical non-ideal realizations. Finally, as an offshoot of our protocol, we present a scheme for secure certification of genuine quantum randomness.
\end{abstract}

\maketitle

\section{Introduction} 
With the rapid development of quantum information science along with its 
multi-faceted applications in  communication and cryptographic protocols, guaranteeing the functioning of quantum devices received from untrusted providers becomes one of the basic requirements for modern quantum technologies. The task of ensuring the proper functioning of the quantum devices can be designed by utilizing the intrinsic features of quantum physics through certification or  verification protocols. Several certification protocols have been designed till date with
the desiderata of efficiency, security and less resource consumption. Tomography \cite{Tomo1,Tomo2,Tomo3}, randomized benchmarking \cite{RB1,RB2,RB3,Eisert_review} and self-testing \cite{my'98,my'04,Supic_review} are notable among them.

Tomography is one of the foremost traditional methodologies for characterizing unknown quantum preparations, measurements, or processes \cite{Tomo1,Tomo2,Tomo3}. However, from an operational perspective, wherein a set of unknown quantum devices are intended to perform information processing or computational tasks, quantum tomography is inadequate. In order to carry out tomography of an unknown quantum device, one first needs to know the dimension as well as the relevant degrees of freedom of the physical system that comprises the device, and accordingly, some other fully characterized quantum devices are essential. For instance, for the tomography of an unknown two-outcome  qubit measurement, at least three completely known qubit preparations are required. 
Therefore,  tomography of quantum state preparation, process or measurement is
a rather resource consuming method. A similar but less resource consuming method is randomized benchmarking which aims to characterize gate errors in an efficient and robust way in terms of the  average overlap  between the physical quantum states, measurements or processes and their ideal counterparts \cite{Eisert_review}. 

Motivated by certain key features of quantum information theory, other ingenious certification methods have  been introduced in recent years. These certifications rely upon various non-classical correlations observed only from the statistics that the devices generate. Moreover, these methods do not require full characterization of any of the devices and, hence, are categorized as device-independent certification protocols. In a fully device-independent scenario, all involved devices are considered as black boxes, thus requiring minimal assumptions on the underlying states and measurements. On the other hand, in a semi device-independent scenario, some  assumptions on the devices are required.  
Apart from their fundamental interests, these certifications have been shown to be immensely useful in many information processing tasks, like quantum key distribution \cite{Acin'07}, secure randomness expansion \cite{Pironio'10}, quantum computation \cite{Vazirani}, and so on. 

The most complete form of device-independent certification, namely, self-testing, employs entanglement and other non-local correlations. With the requirement of space-like separated systems, self-testing provides the optimal possible characterization of entangled systems and quantum measurements without assuming any internal functioning of the devices. Historically, it was first designed in order to certify certain maximally entangled two-qubit states and non-commuting qubit measurements employing the maximum quantum violation of the Bell-CHSH (Bell-Clauser-Horne-Shimony-Holt) inequality \cite{my'98,my'04}. Since then, several other self-testing protocols have been proposed \cite{McKague'12,Yang'13,Bamps'15,Col,Wang'16,supic'181,supic'18}. 

Semi device-independent self-testing protocols have also been investigated \cite{Supic'16,Goswami'18,peng'20,Shrotriya'20,Saha'21,sarkar21} employing Einstein-Podolsky-Rosen steering. Moreover, semi device-independent certification for prepare and measure scenarios have  been  studied \cite{tavakoli'18} with the assumptions on the dimension of the underlying Hilbert space.  
Another class of certification techniques introduced recently, relies on some features of the measurement devices. Out of these, the ones exploiting quantum contextual correlations presume repeatable measurements with certain compatibility relations \cite{kishor'19}, or without any compatibility relations \cite{saha'20}. 

Quantum measurements are one of the most important and key resource in quantum technologies and play a crucial role to reveal the counter-intuitive quantum advantages in  non-classical phenomena. There are several protocols proposed till date to certify various quantum measurements, but most of them either require entanglement \cite{bancal'18,renou'18,kani'17,Mck,Bowles'18}, a costly resource, or need certain assumptions or trust on the measurement devices to be certified \cite{tavakoli'18,Maity'21}. Certification of $d$-outcome measurements (where $d$ is arbitrary) has received attention in a few works \cite{Yang'13,Col} involving scenarios that require a large number of measurements by each of the observers sharing the entangled state.  Recently, certification of $d$-outcome measurements has been proposed based on the Salavrakos-Augusiak-Tura-Wittek-A\'{c}in-Pironio (SATWAP) Bell inequalities \cite{SATWAP}, which involve two measurements on both sides of the shared entangled state \cite{saha}.  However, the above mentioned Bell-nonlocality based protocols require both entanglement as well as space-like separated subsystems in order to ensure loophole-free Bell violation. Therefore, a more efficient certification protocol involving less resources and minimal assumptions for $d$-outcome measurements is in order.

 With the above motivation, here we aim to present a  protocol to certify some specific $d$-outcome quantum measurements (with $d$ being arbitrary) employing the non-classicality of temporal correlations.  Our proposed protocol uniquely (up to some isometry) identifies  which set of measurements is being implemented by an unknown device using measurement statistics and some  partial (not tomographically complete) information. We consider a scenario involving one preparation device and one measurement device. The preparation device produces a maximally mixed state of  an unknown dimension on which the measurement device performs measurements twice in sequence.  In this scenario, we propose a set of temporal inequalities (satisfied by classical devices) pertaining to different values of $d$, containing time separated correlation functions corresponding to successive measurement outcomes. Using quantum violations of these inequalities, we certify a particular set of $d$-outcome quantum measurements without requiring entanglement or any prior knowledge about the dimension of the system. Our scheme relies on certain minimal assumptions  that can be met in practice. We further show that our protocol is robust against non-ideal realizations. Our certification protocol moreover enables us to formulate a scheme for genuine randomness certification.

The rest of the paper is organised as follows. In the next Section, we first present the scenario along with the required assumptions.  We next propose a criterion  certifying that   the measurement effects are projectors. In Section \ref{s4}, we provide our desired set of temporal inequalities along with their sum-of-squares decompositions  under the assumptions considered here. In Section \ref{s5}, we formulate our scheme for certifying $d$-outcome quantum measurements and its robustness analyses. Next, we demonstrate in Section \ref{s6} the framework for secure randomness certification based on our proposed certification protocol. Finally, Section \ref{s7} is reserved for concluding discussions along with some future perspectives. 

\section{Scenario} \label{s2}
Consider the scenario where at first a preparation device $\mathcal{P}$ prepares a state $\rho^{(\mathcal{P})}$ $\in$ $\mathbb{B}(\mathbb{C}^D)$ of an arbitrary dimension $D$.  This prepared state is then subjected to a measurement device $\mathcal{M}$ that performs measurement of the observable $A_i$ upon receiving an input $i$ with $i \in \{ 1, 2, 3, 4\}$. The post-measurement state is again subjected to the same measurement device $\mathcal{M}$ that receives another input $j$ with $j \in \{ 1, 2, 3, 4\}$ and performs measurement of the  observable $A_j$. In each experimental run, $\mathcal{M}$ receives the ordered pair of inputs $(i,j)$ randomly. The outcomes of the measurement of $A_i$ are denoted by $a_i$, where $a_i \in \{0,1, \cdots, d-1\}$. The outcome statistics thus produced are the joint probabilities $p(a_i, a_j |A_i A_j)$ with $i, j, \in \{1,2,3,4\}$ and $a_i, a_j \in \{0,1, \cdots, d-1\}$. Here, $p(a_i, a_j |A_i A_j)$ denotes the joint probability of getting the outcome $a_i$ when the measurement of $A_i$ is performed on the initially prepared state $\rho^{(\mathcal{P})}$, and the 
 outcome of $a_j$ when the measurement of $A_j$ is performed on the post measurement state of $A_i$. This scenario is depicted in Fig. \ref{fig1}.
 
 \begin{figure}[t!]
    \centering
    \includegraphics[scale=0.5]{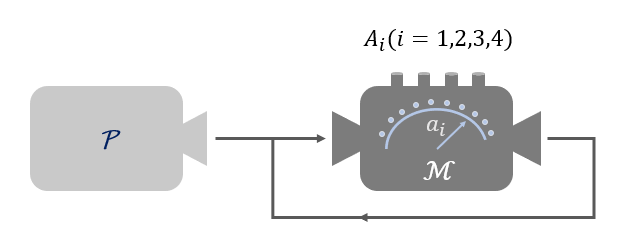}
    \caption{The scenario involves a preparation device $\mathcal{P}$ and a measurement device $\mathcal{M}$ with settings $A_i$ (where $i \in \{1, 2, 3, 4\}$)  which returns outcome $a_i \in \{0, 1, \cdots, d-1\}$. The state  prepared by $\mathcal{P}$ is subjected to $\mathcal{M}$ twice in sequence.}
    \label{fig1}
\end{figure} 

\subsection{Assumptions}

We make the following two assumptions:
\begin{assumption}
The preparation device $\mathcal{P}$ prepares the maximally mixed state $\rho^{(\mathcal{P})} = \mathbbm{1}/D$,  where the dimension $D$ is not required to be known for the realization of the protocol.
\label{ass1}
\end{assumption} 

\begin{assumption}
The measurement device $\mathcal{M}$ always returns the actual post-measurement state, and does not have any memory. 
\label{ass2}
\end{assumption} 

Note that one does not need to know the dimension of the state prepared by $\mathcal{P}$ in order to realize the certification protocol. However, one must trust that $\mathcal{P}$ always prepares maximally mixed state. In other words, an unknown supplier provides a preparation device $\mathcal{P}$ producing an input state  $\rho^{(\mathcal{P})} \in \mathbb{B}(\mathbb{C}^D)$  and a measurement device $\mathcal{M}$ performing measurements of four possible observables acting on the same Hilbert space $\mathbb{C}^D$ (where $D\geq 2$ can have any integer value). We don't know the dimension $D$ or which particular measurements are performed by $\mathcal{M}$, but we trust that the supplier has devised $\mathcal{P}$ in such a way that $\rho^{(\mathcal{P})}$ is a maximally mixed state.
 As we  show later,  such an assumption  or trust  on the preparation device is necessary, else the measurements cannot be certified uniquely using our proposed temporal inequalities. Since, our motto is not to certify the initial state prepared by the preparation device, we can assume the preparation device to be trusted. 

A maximally mixed state can be prepared in the laboratory by subjecting an arbitrary state to a completely depolarizing channel, whose experimental realizations and identification are well-studied \cite{dep1,dep2,dep3,dep4}.  Alternatively, one can prepare a $D$-dimensional maximally mixed state in optical set-up by subjecting a single photon through a multi-branch Mach-Zehnder interferometer \cite{mmzi1,mmzi2} having $D$ number of arms. In each experimental run, one can ensure that the photon passes through a particular arm with unit probability by using specific alignments of a set of mirrors. Hence, this will allow to create mutually orthonormal states in the path degrees of freedom by sending the photons through different arms in different runs. Finally, taking equal mixture of these mutually orthonormal states by using a random number generator to fix the arm through which  the photon will pass in each run, higher dimensional maximally mixed state can be generated.

On the other hand, Assumption \ref{ass2} has two implications. First, no quantum channel is applied on the post-measurement state before the later measurement $A_j$, ensuring that the second measurement acts on the actual post-measurement state. Second,
the specifics of the later measurement depend only on the second input $j$, and independent of the  first input $i$ as well as the outcome of the first measurement $a_i$. 

Although this scenario for certification of measurements requires some assumption on the preparation state and measurement device, it is efficient in the sense that entanglement or other spatial correlations are not necessary for this scheme.

Next, let us derive the expressions of the measurement statistics produced in the aforementioned scenario under Assumptions \ref{ass1}-\ref{ass2}. 

In general, the measurement of $A_i$ (where $i \in \{1, 2, 3, 4\}$)  is represented by the POVM (Positive Operator Valued Measure) as: $A_i \equiv \{M_i^0,  \cdots, M_i^{d-1}  \}$ with $M^{a_i}_i \geq 0$ for all $a_i \in \{0, \cdots, d-1\}$ and  $\sum_{a_i=0}^{d-1} M_i^{a_i} = \mathbbm{1}$. Here, each $M^{a_i}_i$ is called a measurement effect corresponding to outcome $a_i$. The general form of the respective Kraus operators $\{K_i^{a_i}\}$ of the POVM $A_i$ is given by,
\be \label{krausop}
K_i^{a_i}= U_i^{a_i} \sqrt{M_i^{a_1}},
\ee 
where $i \in \{1,2,3,4\},$ $a_i \in \{0, \cdots, d-1\}$, and $U^{a_i}_i$ are some unitary operators. 

Using these notations, the unnormalized post-measurement state $\rho_{a_i}^{i}$, when the outcome $a_i$ is obtained after performing the measurement of $A_i$ on $\rho^{(\mathcal{P})}$, is given by,
\begin{align}
    \rho^{a_i}_{i} = \left(U_i^{a_i} \sqrt{M_i^{a_i}} \right) \rho^{(\mathcal{P})} \left(U_i^{a_i} \sqrt{M_i^{a_i}} \right)^{\dagger}
\end{align}
In the most distrustful scenario, the joint probability distribution, when the measurement of $A_i$ followed by the measurement of $A_j$ is performed, is given by,
\be
p(a_i,a_j|A_i,A_j) = \text{Tr}\Big[ M^{a_j}_{j,i,a_i} \cdot \Lambda_{i,a_i} \left(\rho^{a_i}_{i}\right) \Big] ,
\ee 
where the device can apply some quantum channel $\Lambda_{i,a_i}$ on the post-measurement  quantum state 
and moreover, the specifics of the later measurement, and hence, the measurement effects $\{M^{a_j}_{j,i,a_i}\}$ (where $\sum_{a_j=0}^{d-1} M^{a_j}_{j,i,a_i} = \mathbbm{1}$ for all $a_i,i,j$  and $M^{a_j}_{j,i,a_i} \geq 0$ for all $i,j,a_i,a_j$) may, in general, depend on $i$ and $a_i$. In such case, for any fixed $j$, the associated measurement effects of the later measurement might be different for two different choices of $i$ and/or $a_i$.

However, using Assumption \ref{ass2}, we can consider that the  later measurement measurement effects $\{M^{a_j}_j\}$ are independent of $i$ and/or $a_i$, and there is no quantum channel $\Lambda_{i,a_i}$. Further, using Assumption \ref{ass1}, we obtain a simplified form as
\be
p(a_i,a_j|A_i,A_j) = \frac{\text{Tr} \Big[ M_j^{a_j} U_i^{a_i} M_i^{a_i} \left(U_i^{a_i}\right)^{\dagger}\Big]}{D} .
\label{jointsimple}
\ee 

Due to the fact that the later measurement of $A_j$ cannot influence the outcome statistics of the first measurement $A_i$, one can obtain outcome statistics of the first measurement by taking the appropriate marginals as
\begin{align}
p(a_i | A_i) &= \sum_{a_j = 0}^{d-1} p(a_i, a_j |A_i, A_j) \nonumber \\
&\hspace{0.3cm} \forall a_i \in \{0, \cdots d-1\}, i, j \in \{1, 2, 3, 4\}.
\label{conjtprob}
\end{align}


\subsection{Projectivity of the measurement effects} \label{subIIB}

We would now like to present a lemma that introduces an operational criteria certifying the measurement effects to be projectors.
 
\begin{lemma}
Let the measurement of $A_i$ (where $i \in \{1, 2, 3, 4\}$) satisfies the following condition under Assumptions \ref{ass1}-\ref{ass2}, 
\begin{equation}
p(a_i, a_i | A_i, A_i) = p(a_i | A_i) \, \, \, \, \forall \, \, a_i \in \{0, \cdots, d-1\}.
\label{lemma1e}
\end{equation}
Then all the measurement effects of $A_i$ are mutually orthogonal projectors, that is, $\forall a_i, \widetilde{a_i} \in \{0, \cdots, d-1\}$
\be \label{MM}
M_i^{a_i} M_i^{\widetilde{a_i}} = \delta_{a_i,\widetilde{a_i}} M_i^{a_i},
\ee 
and moreover, each measurement effect is invariant under the respective unitary associated with Kraus operator \eqref{krausop}, that is, $\forall a_i \in \{0, \cdots, d-1\}$ 
\be \label{UMU}
\left(U_i^{a_i}\right) M_i^{a_i} \left(U_i^{a_i}\right)^\dagger = M_i^{a_i}
\ee 
for any choice of $U_i^{a_i}$.
\label{lemma1}
\end{lemma}
A thorough proof of the above Lemma can be found in the Appendix \ref{app1}. For the sake of completeness, here we would like to sketch the outline of the proof.


Without loss of generality, for the POVM $A_i \equiv \{M_i^0,  \cdots, M_i^{d-1}  \}$, one can take $M_i^{a_i} = \sum_{u=0}^{m-1} \lambda_{u} |\psi_u \rangle \langle \psi_u |$, where $1 \leq m \leq D$, $0 < \lambda_{u} \leq 1$ for all $u \in \{0, \cdots, m-1\}$ and $\{ |\psi_0 \rangle,  \cdots, |\psi_{D-1} \rangle \}$ forms an orthonormal basis in $\mathbb{C}^D$. However, if the condition $p(a_i, a_i | A_i, A_i) = p(a_i | A_i)$ is achieved under Assumptions \ref{ass1}-\ref{ass2}, then a detailed calculation implies that  $\lambda_{u} = 1$ for all $u \in \{0,  \cdots, m-1\}$  (see Appendix \ref{app1}). Therefore, $M_i^{a_i} =  \sum_{u=0}^{m-1} |\psi_u \rangle \langle \psi_u |$ and hence, $M_i^{a_i}$ must be a projector. Now, if the above condition holds for all $a_i \in \{0, \cdots, d-1\}$, then Eqs.\eqref{MM}-\eqref{UMU} must hold true. 

Next, we will present a set of temporal inequalities that will be used as a tool for certifying $A_1$, $A_2$, $A_3$, $A_4$.

\section{Temporal Inequality with optimal quantum violation } \label{s4}
Considering the scenario introduced in Section \ref{s2}, we would like to design a set of temporal inequalities which can be used as a witness to certify $A_1$, $A_2$, $A_3$, $A_4$ uniquely (up to some unitary). Nevertheless, it is important to state here that this inequality can be used   in the context of any preparation $\rho^{(\mathcal{P})}$ in the scenario mentioned in section \ref{s2},  without imposing Assumption \ref{ass2}.

Consider now the two-dimensional Fourier transform of the conditional probabilities $p(a_i,a_j|A_i, A_j)$ \cite{gc1,gc2}:
\begin{equation}
\langle A_i^{(k)} \, A_j^{(l)}\rangle = \sum_{a_i,a_j=0}^{d-1} \omega^{a_i k  + a_j l} p(a_i,a_j|A_i, A_j),
\label{exp1}
\end{equation}
where $\omega$ is the $d$-th root of unity i.e., $\omega = \exp (2 \pi \mathbbm{i}/d )$; $k, l \in \{0,  \cdots, d-1\}$; $i, j \in \{1,2,3,4\}$. 
Here, $\{A_x^{(z)}\}$ for each $x \in \{1,2,3,4\}$  are the Fourier transformed operators  defined as \cite{Saha'21}
\begin{equation}
	A_x^{(z)} = \sum_{a_x=0}^{d-1} \omega^{a_x z} M_x^{a_x} \hspace{0.3cm} \text{with} \, \, \, z=0, \cdots, d-1,
	\label{obsnew}
\end{equation}
where $ A_x \equiv \{M_x^{a_x} | M_x^{a_x} \geq 0 \, \, \forall \, a_x, \, \sum_{a_x} M_x^{a_x} = \mathbbm{1}\}$ as mentioned earlier. Each $A_x^{(z)}$ can be termed as a generalized observable. 

It can be checked that for all $z \in \{0,  \cdots, d-1\}$ and $x\in \{1,2,3,4\}$
\begin{align}
	&A_x^{(z)^\dagger} = A_x^{(d-z)} = A_x^{(-z)}, \nonumber \\
	& A_x^{(z)^\dagger} A_x^{(z)} \leq \mathbbm{1}, \nonumber \\
	&A_x^{(0)} = \mathbbm{1}.
	\label{conditionforpovm}
\end{align} 
Importantly,  as a special case, when $M_x^{a_x} M_x^{\widetilde{a_x}} = \delta_{a_x, \widetilde{a_x}} M_x^{a_x}$ for all $a_x, \widetilde{a_x} \in \{0, \cdots, d-1\}$, i.e., all the POVM effects $\{M_x^{a_x}\}$ are mutually orthogonal projectors for each $x$, then we can define $A_x : = A^{(1)}_x = \sum_{a_x=0}^{d-1} \omega^{a_x} \,  \Pi_x^{a_x}$ for each $x$, where $\{\Pi_x^{a_x}\}$ are the respective projectors. In this case, $\{A_x^{(z)} \}$ for each $x \in \{1,2,3,4\}$  are collections of unitary operators with eigenvalues  $\omega^i$ $(i = 0, \cdots ,d-1)$ defined as 
\begin{equation}
A_x^{(z)} = \sum_{a_x=0}^{d-1} \omega^{a_x z} \Pi_x^{a_x} \hspace{0.3cm} \text{with} \, \, \, z=0, \cdots, d-1.
\label{obs}
\end{equation}	
It is not difficult to see from the above relation (\ref{obs}) that $A_x^{(z)}$ is simply the $z$-th power of $A_x$. Thus, in what
follows we use the notation $A_x^{(z)}$ or $A_x^z$  interchangeably when the effect operators associated with the measurement of $A_x$ are mutually orthogonal projectors.

\subsection{Temporal inequalities} 
Next, we propose the following set of temporal inequalities,
\begin{align}
\tau_d = &\sum_{k=1}^{d-1} \Bigg[ a_k \left\langle A_1^{(k)} \, A_3^{(d-k)} \right\rangle + a_k^{*} \omega^k \left\langle A_1^{(k)}  \, A_4^{(d-k)} \right\rangle \nonumber \\
& \, \, \, \, \, \, \, + a_k^{*} \left\langle A_2^{(k)} \, A_3^{(d-k)} \right\rangle  + a_k \left\langle A_2^{(k)} \, A_4^{(d-k)} \right\rangle \nonumber \\
& \, \, \, \, \, \, \, + 
a_k \left\langle  A_3^{(d-k)} \, A_1^{(k)} \right\rangle + a_k^{*} \omega^k \left\langle A_4^{(d-k)} \, A_1^{(k)}  \right\rangle \nonumber \\
& \, \, \, \, \, \, \,+ a_k^{*} \left\langle A_3^{(d-k)} \,  A_2^{(k)} \right\rangle + a_k \left\langle A_4^{(d-k)} \, A_2^{(k)}  \right\rangle \Bigg] \nonumber \\
& \leq C_d,
\label{temp}
\end{align} 
where $a_k = \dfrac{1- \mathbbm{i}}{2} \exp\Bigg(\dfrac{\pi \mathbbm{i} k}{2d} \Bigg)$ and 
\begin{align}
C_d =  3 \cot \Bigg(\dfrac{\pi}{4d} \Bigg) - \cot \Bigg( \dfrac{3\pi}{4d} \Bigg) - 4
\label{cbound}
\end{align}
 are the classical upper bounds of the temporal expressions $\tau_d$. These classical bounds are derived in the Appendix \ref{app2}.  In particular, the above temporal inequalities are not only satisfied by classical physics, but also satisfied by any theory consistent with the concept of ``macrorealism" \cite{lgi} which is the conjunction of the following two assumptions: \textit{(i) Realism:} At any instant, irrespective of any measurement, a system is definitely in any one of the available states such that all its observable properties have definite values. \textit{(ii) Noninvasive measurability:} It is possible, in principle, to determine which of the states the system is in, without affecting the state itself or the system's subsequent evolution. This notion of ``macrorealism" is one of the central concepts underpinning the classical world view.   The temporal inequalities (\ref{temp}) can be considered as  generalized versions of the Leggett-Garg inequality involving  measurements with arbitrary number of outcomes whereas the original Leggett-Garg inequality \cite{lgi} consists of binary outcome measurements. Temporal quantum correlations are inconsistent with macrorealism \cite{lgi} and, hence, it is possible to violate the above inequalities by quantum mechanical predictions. A point to be stressed here is that (\ref{temp}) represents different temporal inequalities for different $d$.

It should be mentioned that although the structure of the inequalities (\ref{temp}) is somewhat similar to the structure of the SATWAP Bell inequalities involving  measurements with arbitrary number of outcomes \cite{SATWAP}, these two sets of inequalities are conceptually different. The SATWAP inequalities consist of correlation functions between measurement outcomes of two spatially separated systems \cite{SATWAP}. On the other hand, the inequalities (\ref{temp}) proposed here contain correlation functions pertaining to a single system on which different measurements are performed sequentially.  Moreover, the above temporal inequalities (\ref{temp}) have twice more terms than the number of terms appearing in SATWAP inequalities.  Note further  that although the above inequalities (\ref{temp}) are expressed in the Fourier transformed space, these can also be expressed as linear functions of $p(a_i,a_j|A_i,A_j)$ with real coefficients and always take real values as shown in the Appendix \ref{app2}.

Up to now, the details of  the Fourier transformed expectation values and  the temporal inequalities (\ref{temp}) are discussed for an arbitrary preparation $\rho^{(\mathcal{P})}$. The following discussions will be applicable when  Assumptions \ref{ass1}-\ref{ass2} are considered.

Before proceeding, let us point out that  the optimal quantum violations of the temporal inequalities (\ref{temp}) under Assumptions \ref{ass1}-\ref{ass2} and condition (\ref{lemma1e}) are relevant for the certification scheme presented here.

Now, suppose that under Assumption \ref{ass1}-\ref{ass2}, the condition (\ref{lemma1e}) is satisfied  by each of the four  observables $A_1$, $A_2$, $A_3$ and $A_4$.  Here, $A_i =  A_i^{(1)}$ for all $i \in \{1,2,3,4\}$ with $A_i^{(1)}$ being defined in Eq.(\ref{obs}).  Then, it follows from  Lemma \ref{lemma1} that each of the measurement of $A_i$ can be represented as $A_i \equiv \{\Pi^{a_i}_i\}$, where $\Pi_i^{a_i} \Pi_i^{\widetilde{a_i}} = \delta_{a_i, \widetilde{a_i}} \Pi_i^{a_i}$ for all $a_i, \widetilde{a_i} \in \{0, \cdots, d-1\}$. Hence, the operator $A_i^{(k)}$  is  unitary with $A_i^{(k)} = A_i^k$ for all $i \in \{1,2,3,4\}$ and for all $k \in \{0, \cdots, d-1\}$.
Moreover, due to \eqref{UMU}, the joint probability expressed in Eq. \eqref{jointsimple} further reduces to,
\begin{align}
p(a_i,a_j|A_i,A_j) =&\text{Tr} \Big[ \Pi_j^{a_j} \,\Pi_i^{a_i} \, \rho^{(\mathcal{P})} \Big] \nonumber \\
&
\, \, \, \forall \, i, j \in \{1, 2, 3, 4\}, \nonumber \\
& \, \, \,  \text{and} \, \, \,  \forall \, a_i, a_j \in \{0,  \cdots, d-1\} ,
\label{pro1newnew}
\end{align} 
in which $\rho^{(\mathcal{P})} = \mathbbm{1}/D$. 

Consequently, using Eqs.(\ref{exp1}) and (\ref{obs}), we can write for all $k, l =1, \cdots, d-1$ and for all $i, j \in \{1, 2, 3, 4\}$,
\begin{align}
 \langle A_i^{(k)} \, A_j^{(l)}\rangle &= \text{Tr}\Big[A_i^{k} \, A_j^{l} \, \rho^{(\mathcal{P})} \Big] \nonumber \\ 
 &=\text{Tr}\Big[A_i^{(k)} \, A_j^{(l)} \, \rho^{(\mathcal{P})} \Big]. \nonumber
\end{align} 
Hence,  under Assumptions \ref{ass1}-\ref{ass2} and condition \eqref{lemma1e}, the left hand side of the temporal inequality (\ref{temp}) can be expressed as
\begin{align}
\tau_d  \left( \rho^{(\mathcal{P})} = \mathbbm{1} / D\right) = \text{Tr} \big[ \rho^{(\mathcal{P})} \, \hat{\beta}_{\tau_d} \big] ,
\label{temp2}
\end{align}
where the operator $\hat{\beta}_{\tau_d}$ is given by,
\begin{align}
\hat{\beta}_{\tau_d} =  \sum_{k=1}^{d-1} & \Big[  a_k  A_1^{(k)} \, A_3^{(d-k)}  + a_k^{*} \omega^k  A_1^{(k)}  \, A_4^{(d-k)}  \nonumber \\
&+ a_k^{*}  A_2^{(k)} \, A_3^{(d-k)}   + a_k  A_2^{(k)} \, A_4^{(d-k)}  \nonumber \\
&+ a_k   A_3^{(d-k)} \, A_1^{(k)}  + a_k^{*} \omega^k A_4^{(d-k)} \, A_1^{(k)}   \nonumber \\
& + a_k^{*}  A_3^{(d-k)} \,  A_2^{(k)}  + a_k  A_{(4)}^{(d-k)} \, A_2^{(k)}   \Big] .
\label{temp3}
\end{align} 
It should be noted here that (\ref{temp2}) is valid under the Assumptions \ref{ass1}-\ref{ass2} and when each of the four observables $A_i$ with $i \in \{1,2,3,4\}$ satisfies the condition (\ref{lemma1e}). In other words, the condition (\ref{temp2}) may not be true in general.

\subsection{Sum-of-squares decomposition with maximally mixed state}\label{sossection}
Now, let us derive the sum-of-squares decompositions of the temporal inequalities \eqref{temp}  under Assumptions \ref{ass1}-\ref{ass2}, when each of the four observables $A_i$ with $i \in \{1,2,3,4\}$ satisfies condition (\ref{lemma1e}). For this purpose,  let us first define
\begin{align}\label{defBi}
&B_1^{(k)} = a_k A_3^{(-k)} + a_k^{*} \omega^k A_4^{(-k)}, \nonumber \\
&B_2^{(k)} = a_k^{*} A_3^{(-k)} + a_k  A_4^{(-k)}.
\end{align} 
Note here that $a_{d-k} = a_k^{*}$, and therefore $B_x^{(d-k)} = \Big[B_x^{(k)} \Big]^{\dagger}$ for any $k=1, \cdots, d-1$ and $x=1,2$. 

Since, $(A_i^{k})^{\dagger} = A_i^{d-k}$ and $A_i^{k}$ is unitary for all $i \in \{1, 2, 3, 4\}$ and for all $k \in \{1, \cdots, d-1\}$, we have $A_i^{-k} = A_i^{d-k}$ for all $i \in \{1, 2, 3, 4\}$ and for all $k \in \{1, \cdots, d-1\}$. Since $A_x^{(k)} = A_x^k$ for all $x \in \{1,2\}$ and for all $k \in \{1, \cdots, d-1\}$, it follows that 
\begin{align}
\hat{\beta}_{\tau_d} = \sum_{k=1}^{d-1} \Big[&  A_1^{(k)} \, B_1^{(k)}  +  A_2^{(k)} \, B_2^{(k)}  \nonumber \\
&+  B_1^{(k)} \, A_1^{(k)}   +  B_2^{(k)} \, A_2^{(k)}  \Big]. \nonumber
\end{align} 
Let us also define,
\begin{align}
P_x^{(k)} = \mathbbm{1} - A_x^{(k)} \, B_x^{(k)}& \hspace{0.3cm} \forall x \in \{1, 2\}, \nonumber \\
 &\forall k \in \{1, \cdots, d-1\}. 
\label{pxk}
\end{align} 
Using these relations, we have
\begin{align}
&\sum_{k=1}^{d-1} \sum_{x=1}^{2} \Big[ \Big(P_x^{(k)}\Big)^{\dagger} \, \Big(P_x^{(k)}\Big) \Big] \nonumber \\
&= \sum_{k=1}^{d-1} \Big[ 2 \, \mathbbm{1} - B_1^{(d-k)} A_1^{d-k} - A_1^{k} B_1^{(k)} \nonumber \\
& \hspace{0.5cm}+ \Big(B_1^{(k)} \Big)^{\dagger} \Big(B_1^{(k)} \Big) - B_2^{(d-k)} A_2^{d-k} \nonumber \\
& \hspace{0.5cm}  - A_2^{k} B_2^{(k)} + \Big(B_2^{(k)} \Big)^{\dagger} \Big(B_2^{(k)} \Big) \Big] .
\label{sos1}
\end{align}
One can verify that
\begin{equation}
\sum_{k=1}^{d-1} B_x^{(d-k)} A_x^{d-k} = \sum_{k=1}^{d-1} B_x^{(k)} A_x^{k} \hspace{0.5cm} \forall x = 1, 2.
\label{sos2}
\end{equation}
Now, incorporating the fact that $\Big(A_3^{d-k}\Big)^{\dagger} \Big(A_3^{d-k}\Big) = \Big(A_4^{d-k}\Big)^{\dagger} \Big(A_4^{d-k}\Big) = \mathbbm{1}$ and $(a_k)^2 (\omega^k)^{*} + (a_k^{*})^2 = (a_k^{*})^2 \omega^k + (a_k)^2 = 0$ for all $k \in \{1, \cdots, d-1\}$ one can evaluate that
\begin{align}
\Big(B_1^{(k)} \Big)^{\dagger} \Big(B_1^{(k)} \Big) + \Big(B_2^{(k)} \Big)^{\dagger} \Big(B_2^{(k)} \Big) = 2 \, \mathbbm{1} \hspace{0.2cm} \forall k.
\label{sos4}
\end{align}

Using Eqs.(\ref{sos1}) (\ref{sos2}), (\ref{sos4}), we have
\begin{align}
\sum_{k=1}^{d-1} \sum_{x=1}^{2} \Big[ \Big(P_x^{(k)}\Big)^{\dagger} \, \Big(P_x^{(k)}\Big) \Big] = 4 (d-1) \, \mathbbm{1} - \hat{\beta}_{\tau_d}.
\label{sos5}
\end{align}
The left hand side of Eq.(\ref{sos5}) is the  sum-of-squares decomposition and it is the sum of positive operators. Hence, we have
\begin{equation}
\text{Tr}\Big[ \rho^{(\mathcal{P})} \big(4 (d-1) \, \mathbbm{1} - \hat{\beta}_{\tau_d} \big) \Big] \geq 0. \nonumber
\end{equation}
The above inequality implies that
\begin{equation} \label{maxqvalue}
\text{Tr}\big[ \rho^{(\mathcal{P})} \, \hat{\beta}_{\tau_d} \big] \leq 4(d-1).
\end{equation}
Hence, the upper bounds of the quantum magnitudes of $\tau_d$ for all $d \geq 2$ are $4(d-1)$  under Assumptions \ref{ass1}-\ref{ass2} and condition (\ref{lemma1e}). It can be checked that $4(d-1) > C_d$ for all $d \geq 2$ \cite{SATWAP}, where $C_d$ given by Eq.(\ref{cbound}) are the classical upper bounds of the temporal expressions $\tau_d$.  Again, it should be noted here that $4(d-1)$ may not be the optimal upper bounds of quantum violations of the temporal inequalities (\ref{temp}) in general.

\subsection{Binary outcome ($d=2$) case} \label{sec:4C}

In the case of $d=2$, the temporal inequality (\ref{temp}) reduces to 
\begin{align}
\tau_2  &=   \frac{1}{\sqrt{2}} \Big[ \langle A_1  A_3 \rangle- \langle A_1   A_4 \rangle +  \langle A_2  A_3 \rangle +  \langle A_2  A_4 \rangle  \nonumber \\
& \, \, \, \, \, \, +  \langle  A_3  A_1 \rangle- \langle A_4  A_1  \rangle +  \langle A_3   A_2 \rangle +  \langle A_4  A_2  \rangle \Big] \nonumber \\
&\leq 2 \sqrt{2},
\label{tid2}
\end{align} 
where the measurement of $A_i$ (with $i \in \{1, 2, 3, 4\}$) has two possible outcomes denotes by $a_i \in \{0, 1\}$. This turns out to be the symmetrized version of the temporal inequality proposed in \cite{entintime}.  Interestingly, in this particular binary outcome scenario, the expression \eqref{temp2} for the quantum value of \eqref{tid2}  holds for any general input state, that is, for any $\rho^{(\mathcal{P})} \in \mathbb{B}(\mathbb{C}^D)$ (where $D$ is arbitrary),
\begin{align}
\tau_{2} \left( \rho^{(\mathcal{P})} \right) = \text{Tr} \big[ \rho^{(\mathcal{P})} \, \hat{\beta}_{\tau_2} \big]  \, \, \, \forall \, \, \rho^{(\mathcal{P})} \in \mathbb{B}(\mathbb{C}^D),
\label{temp22}
\end{align}
where the operator $\hat{\beta}_{\tau_2}$ is given by,
\begin{align}
   \hat{\beta}_{\tau_2}  =&  A_1 \, A_3 -  A_1  \, A_4 +  A_2 \, A_3  +  A_2 \, A_4 \nonumber \\
    & +  A_3 \, A_1 - A_4 \, A_1  +  A_3 \,  A_2  +  A_4 \, A_2  .
    \label{betatau2}
\end{align} 
To see this,  note that there is no restriction on the dimension $D$ of the measurements. Therefore, the Naimark's dilation theorem allows us to consider these measurements to be projective given by,
\begin{align}
    A_i = \Pi_i^{0} - \Pi_i^{1} \hspace{0.3cm} \forall i \in \{1, 2, 3, 4\},
    \nonumber
\end{align}
where $\{\Pi_i^{0} , \Pi_i^{1} \}$ are projectors acting on $\mathbb{C}^D$ with 
\begin{align}
    \Pi_i^{a_i} = \frac{\mathbbm{1} +(-1)^{a_i} A_i}{2} \, \, \forall \, a_i \in \{0,1\} .
    \label{md2e4}
\end{align} 
The joint probability for any state $\rho^{(\mathcal{P})}$ is given by, 
\begin{align}
    p(a_i,a_j|A_i, A_j) = \text{Tr}\Big[ \Pi_j^{a_j} \Pi_i^{a_i} \rho^{(\mathcal{P})} \Pi_i^{a_i} \Big] 
    \label{md2e3}
\end{align}
$\forall i, j \in \{1, 2, 3, 4\}$ and $\forall a_i, a_j \in \{0,1\}$. 
Using the expression of correlation function
\begin{align}
\langle A_i \, A_j \rangle = \sum_{a_i,a_j=0}^{1} (-1)^{a_i  + a_j} p(a_i,a_j|A_i, A_j),
\label{d2e1}
\end{align}
and Eqs. (\ref{md2e4})-(\ref{md2e3}), we get the following
\begin{align}
 \langle A_i \, A_j \rangle + \langle A_j \, A_j \rangle =&  \text{Tr}\Big[\left(A_i A_j + A_j A_i \right) \rho^{(\mathcal{P})} \Big] \nonumber \\ 
 &\, \, \, \, \, \, \, \, \, \, \forall \, \, \rho^{(\mathcal{P})} \in \mathbb{B}(\mathbb{C}^D).    
\end{align}
By considering pairs of terms in the temporal expression \eqref{tid2}, one can readily verify that the associated quantum operator is given by \eqref{betatau2}.
Subsequently, the sum-of-squares decomposition \eqref{sos5} for $d=2$ and the quantum upper bound $4$ of \eqref{tid2} hold true for any general preparation $\rho^{(\mathcal{P})}$.

\section{Certification of quantum measurements} \label{s5}

In this section, we would like to state our main results for certifying the $d$-outcome quantum measurements employing the set of temporal inequalities proposed by us in the previous section along with the sum-of-squares decompositions (\ref{sos5}). Before proceeding further, let us first verify that the measurements cannot be certified uniquely using the temporal inequalities (\ref{temp}) if we do not make any assumption on the state preparation. Here, uniqueness encompasses the unitary freedom of the measurements. 

\begin{lemma} 
There exist at least two sets of preparations and binary-outcome projective measurements  for which the magnitude of quantum violation of the temporal inequality (\ref{tid2})  is $4$, but the measurements in these two sets are not unitarily connected. In other words, uniqueness of the measurements cannot be shown employing binary outcome temporal inequality (\ref{tid2})  without any assumption on the preparation.
\label{lemma2}
\end{lemma}

\begin{proof}

Let us first note that the optimal quantum value $4$ of \eqref{tid2} holds for arbitrary preparation $\rho^{(\mathcal{P})} \in \mathbb{C}^D$. Now, we show that there exist two different sets of preparations and projective measurements  such that the measurements in these two sets are not connected unitarily although  $\tau_2 = 4$ is achieved  by both of these two sets. Hence, for arbitrary preparations, the measurements  cannot be certified uniquely using the inequality (\ref{tid2}) without any assumption on the preparation device. The explicit examples of such two different sets of preparations and projective  measurements  are given in the Appendix \ref{app3}.
\end{proof}

The above lemma implies that some assumption on the preparation device is necessary in order to certify the measurements uniquely using the quantum violation of the temporal inequality (\ref{temp}). Therefore, we have considered Assumption \ref{ass1} which states that the preparation device produces a maximally mixed state of  an unknown dimension. 

Before proceeding, let us define the $d$-dimensional generalization of the $\sigma_z$-Pauli matrix in the
standard basis given by,
\begin{align}
    Z_d = \sum_{i=0}^{d-1} \omega^i |i \rangle \langle i |.
\end{align}
let us also introduce the following $d$-dimensional unitary observable,
\begin{align}
    T_d =& \sum_{i=0}^{d-1} \omega^{i+\frac{1}{2}} |i \rangle \langle i | \nonumber \\
    & \, \, \, \, - \frac{2}{d} \sum_{i,j=0}^{d-1} (-1)^{\delta_{i,0} + \delta_{j,0}} \omega^{\frac{i + j + 1}{2}} |i \rangle \langle j |,
    \label{exptd}
\end{align}
where $\delta_{i,j}$ is the Kronecker delta function. It can be checked that $Z_d$ and $T_d$ are unitary with eigenvalues $\omega^i$
($i = 0, \cdots, d - 1$)  \cite{saha}.

With these, we present a theorem  providing the certification of a specific set of quantum measurements. This theorem states that one can certify the four observables $A_1,A_2,A_3,A_4,$ uniquely (up to some unitary freedom) using the condition (\ref{lemma1e})  and  the temporal inequalities (\ref{temp}). Hence, one can indeed certify some particular quantum measurements having arbitrary number of outcomes  under certain assumptions employing temporal quantum correlations without using entanglement. 

\begin{thm}\label{theo}
Suppose that  in the scenario considered by us with any fixed value  of $d$ under Assumptions \ref{ass1}-\ref{ass2},  the condition (\ref{lemma1e}) is satisfied by each of the four observables $A_1,A_2,A_3,A_4,$ and the magnitude of quantum violation of the temporal inequality (\ref{temp})  is $4(d-1)$. Then, for any $d$, we have  $\mathbb{C}^D = \mathbb{C}^d \otimes \mathcal{H}' $ with some
auxiliary Hilbert space $\mathcal{H}'$ of unknown but finite dimension. Further, there exists a  unitary transformation $U : \mathbb{C}^d \otimes \mathcal{H}'  \rightarrow \mathbb{C}^d \otimes \mathcal{H}'$, such that
\begin{align} 
U A_1 U^{\dagger} &= Z_d \otimes  \mathbbm{1}_{\mathcal{H}'},  \nonumber \\
U A_2 U^{\dagger} &= T_d \otimes  \mathbbm{1}_{\mathcal{H}'}, \nonumber \\
U A_3 U^{\dagger} &= \Big(a_1^{*} Z_d + 2 (a_1^*)^3 T_d \Big) \otimes  \mathbbm{1}_{\mathcal{H}'},  \nonumber \\
U A_4 U^{\dagger} &= \Big(a_1 Z_d - a_1^* T_d \Big) \otimes  \mathbbm{1}_{\mathcal{H}'},  
\label{certifiedmeasurements}
\end{align}
where $a_1 = \dfrac{1-\mathbbm{i}}{2} \omega^{\frac{1}{4}}$, $Z_d$ and $T_d$ are defined earlier and $\mathbbm{1}_{\mathcal{H}'}$ is the identity matrix acting on $\mathcal{H}'$.
\end{thm}

\begin{proof}
Under Assumptions \ref{ass1} and \ref{ass2}, we can take the measurement effects of each of the observables $A_i$ with $i \in \{1,2,3,4\}$ to be mutually orthogonal projectors as each of these   observables satisfies the condition (\ref{lemma1e}). This follows from Lemma \ref{lemma1}.  Hence, $A_x^{(k)}$  is  unitary operator with $A_x^{(k)} = A_x^k$ for all $x \in \{1,2,3,4\}$ and for all $k \in \{0, \cdots, d-1\}$. Also, we can consider that the condition (\ref{UMU}) holds true.

Next, since the magnitude of quantum violation of the temporal inequality (\ref{temp})  is $ 4(d-1)$  under Assumptions \ref{ass1}-\ref{ass2}, using Eq.(\ref{sos5}), we can write
\begin{align}
\text{Tr} \Bigg[ \rho^{(\mathcal{P})} \, \Bigg\{ \sum_{k=1}^{d-1} \sum_{x=1}^{2}  \Big(P_x^{(k)}\Big)^{\dagger} \, \Big(P_x^{(k)}\Big) \Bigg\} \Bigg] = 0. \nonumber
\end{align}
Since $\Big(P_x^{(k)}\Big)^{\dagger} \, \Big(P_x^{(k)}\Big) \geq 0$ for all $k \in \{1, \cdots, d-1\}$ and for all $x \in \{1,2\}$, the above equation implies that
\begin{align} 
\text{Tr} \Bigg[ \rho^{(\mathcal{P})} \,  \Big(P_x^{(k)}\Big)^{\dagger} \, \Big(P_x^{(k)}\Big)  &\Bigg] = 0 \nonumber \\
& \forall \, k \in \{1, \cdots, d-1\} \nonumber \\
&\text{and} \, \, \, \forall \, x \in \{1,2\}.
\label{p2} 
\end{align}
As mentioned earlier, $\rho^{(\mathcal{P})} = \sum_{u=0}^{D-1} \frac{1}{D} | \xi_u \rangle \langle \xi_u |$
with $\{|\xi_u \rangle \}$ being an arbitrary orthonormal basis in $\mathbb{C}^D$. Hence, from Eq.(\ref{p2}), we have
\begin{align}
\Big(P_x^{(k)}\Big) &| \xi_u \rangle = 0, \hspace{0.5cm}  \langle \xi_u |\Big(P_x^{(k)}\Big)^{\dagger} = 0 \nonumber \\
&\forall \, x \in \{1,2\}, \, \, \forall \, k \in \{1, \cdots, d-1\}, \nonumber \\
& \forall \, u \in \{0, \cdots, D-1\}.  \nonumber
\end{align}
We can, therefore, write that
\begin{align}
\Big(P_x^{(k)}\Big)& \rho^{(\mathcal{P})} = 0, \hspace{0.5cm} \rho^{(\mathcal{P})} \Big(P_x^{(k)}\Big)^{\dagger} = 0 \nonumber \\
& \forall \, x \in \{1,2\}, \, \, \forall \, k \in \{1, \cdots, d-1\}.
\label{p4}
\end{align}
Taking $\rho^{(\mathcal{P})} = \mathbbm{1}/D$ and using Eqs.(\ref{pxk}) and (\ref{p4}), we get
\begin{align}
A_x^{k} \, B_x^{(k)} &= \mathbbm{1}, \hspace{0.5cm} \Big(B_x^{(k)}\Big)^{\dagger} \, \Big(A_x^{k}\Big)^{\dagger} = \mathbbm{1} \nonumber \\ &\forall \, x \in \{1,2\}, \, \, \forall \, k \in \{1, \cdots, d-1\}.  
\label{p5n}
\end{align}
Using the above two equations along with the condition $\Big(A_x^{k}\Big)^{\dagger} \, \Big(A_x^{k}\Big) =\mathbbm{1}$, one has
\begin{align}
&\Big(B_x^{(k)}\Big)^{\dagger} \, \Big(B_x^{(k)} \Big) = \Big(B_x^{(d-k)}\Big) \, \Big(B_x^{(k)} \Big) =  \mathbbm{1} \nonumber \\
&\hspace{0.9cm} \forall \, x \in \{1,2\}, \, \, \forall \, k \in \{1, \cdots, d-1\}.
\label{p9nnnn}
\end{align}
Now, taking $x=1$, the above condition leads to,
\begin{align}
A_3^{k} \, A_4^{-k} = \omega^{-k} A_4^{k} \, A_3^{-k} \hspace{0.5cm}  \forall \, k \in \{1, \cdots, d-1\}.
\label{p10}
\end{align}
Due to the fact that $A_i^{d}= \mathbbm{1}$ for all $i \in \{1,2,3,4\}$, the above relation (\ref{p10}) can be extended to any integer $k \in \mathbb{Z}$.

Next, Eq.(\ref{p5n}) also implies that 
$B_x^{(k)} = (A_x^k)^{\dagger} =  (A_x^{\dagger})^k$ for all $x \in \{1, 2\}$ and $k \in \{1, \cdots, d-1\}$. Moreover, with $k=1$, Eq.(\ref{p5n}) implies that  $B_x^{(1)} = A_x^{\dagger}$ for all $x \in \{1, 2\}$. Hence, we get
\begin{align}
B_x^{(k)} &=   (A_x^{\dagger})^{k} \nonumber \\
&= (B_x^{(1)})^{k} \hspace{0.15cm} \forall \, x \in \{1,2\}, \,  \forall \, k \in \{1, \cdots, d-1\}.
\label{p11}
\end{align}

Eqs.(\ref{p5n}), (\ref{p9nnnn}), (\ref{p10}) and (\ref{p11}) are sufficient to characterize $A_1$, $A_2$ $A_3$ and $A_4$. In fact, it can be shown that if the unitary observables $A_3$ and $A_4$ satisfy the conditions (\ref{p9nnnn}), (\ref{p10}) and (\ref{p11}), then we can draw the following two conclusions (the calculations are the same as presented in the  Supplementary Information  of \cite{saha}):

\begin{itemize}
    \item The dimension $D$ is a multiple of the number of outcomes $d$, meaning that
\begin{align}
\mathbb{C}^D = \mathbb{C}^d \otimes \mathcal{H}' ,\nonumber
\end{align}
where  $\mathcal{H}'$ is some
auxiliary Hilbert space of finite dimension $D/d$ which is unknown as  $D$ is unknown.

\item There exists a  unitary transformation $\widetilde{U}: \mathbb{C}^d \otimes \mathcal{H}'  \rightarrow \mathbb{C}^d \otimes \mathcal{H}'$, such that
\begin{align}
\widetilde{U} A_3 \widetilde{U}^{\dagger} &= Z_d \otimes  \mathbbm{1}_{\mathcal{H}'}, \label{p161} \\
\widetilde{U} A_4 \widetilde{U}^{\dagger} &= T_d \otimes  \mathbbm{1}_{\mathcal{H}'},
\label{p16}  
\end{align}
where $Z_d$ and $T_d$ are defined earlier.
\end{itemize}

Now, it can be proved that there exists a  unitary transformation $W =  : \mathbb{C}^d  \rightarrow \mathbb{C}^d$, such that
\begin{align}
W Z_d W^{\dagger} &= \Big(a_1^{*} Z_d + 2 (a_1^*)^3 T_d \Big),  \label{wun1} \\
W T_d W^{\dagger} &= \Big(a_1 Z_d - a_1^* T_d \Big), 
\label{wun2}
\end{align}
where $a_1 = \dfrac{1- \mathbbm{i}}{2} \omega^{\frac{1}{4}}$. The existence of such a unitary $W$ is proved in the  Supplementary Information  of \cite{saha}.

Therefore, using Eqs.(\ref{p161})-(\ref{wun2}), we can conclude that there exists a unitary transformation $U =  (W \otimes \mathbbm{1}_{\mathcal{H}'}) \, \widetilde{U}: \mathbb{C}^d \otimes \mathcal{H}'  \rightarrow \mathbb{C}^d \otimes \mathcal{H}'$, such that 
\begin{align}
U A_3 U^{\dagger} &= \Big(a_1^{*} Z_d + 2 (a_1^*)^3 T_d \Big) \otimes  \mathbbm{1}_{\mathcal{H}'}, \label{p182}  \\
U A_4 U^{\dagger} &= \Big(a_1 Z_d - a_1^* T_d \Big) \otimes  \mathbbm{1}_{\mathcal{H}'}.
\label{p18} 
\end{align}
Now, Eqs. (\ref{p11}), (\ref{p182}) and (\ref{p18}) imply that 
\begin{align}
U B^{(k)}_1 U^{\dagger} &= \big(Z^{\dagger}_d \big)^{k} \otimes  \mathbbm{1}_{\mathcal{H}'}, \label{p19}  \\
U B^{(k)}_2 U^{\dagger} &=\big(T^{\dagger}_d \big)^{k}  \otimes  \mathbbm{1}_{\mathcal{H}'}.
\label{p192} 
\end{align}
In particular, if $A_3$ and $A_4$ are transformed by applying the unitary $U$, then $B^{(1)}_1 = Z^{\dagger}_d$ and $B^{(1)}_2 = T^{\dagger}_d$. 

Next, Eq.(\ref{p5n}) implies that $A_x^k \, B_x^{(k)} = \mathbbm{1}$ for all $x \in \{1, 2\}$ and $k \in \{1, \cdots, d-1\}$. Hence,  with $k=1$, we get
\begin{align}
A_x = \big( B_x^{(1)} \big)^{\dagger} \hspace{0.9cm} \forall \, x \in \{1,2\}.
\label{p20}
\end{align}
Hence, Eqs.(\ref{p19})-(\ref{p20}) imply that there exists a  unitary transformation $U : \mathbb{C}^d \otimes \mathcal{H}'  \rightarrow \mathbb{C}^d \otimes \mathcal{H}'$, such that
\begin{align}
U A_1 U^{\dagger} &= Z_d \otimes  \mathbbm{1}_{\mathcal{H}'},  \label{p211}  \\
U A_2 U^{\dagger} &= T_d \otimes  \mathbbm{1}_{\mathcal{H}'}.
\label{p21}  
\end{align}
Thus, Eqs.(\ref{p182})-(\ref{p18}) together with Eqs.(\ref{p211})-(\ref{p21}) complete the proof.
\end{proof}

 The above Theorem also implies that $4(d-1)$ is the tight upper bound of the temporal inequality \eqref{temp}  for any fixed $d$ under Assumptions \ref{ass1}-\ref{ass2} when the condition (\ref{lemma1e}) is satisfied by each of the four observables $A_1,A_2,A_3,A_4$. Let us also remark that the certified operators $(a_1^{*} Z_d + 2 (a_1^*)^3 T_d)$ and $(a_1 Z_d - a_1^* T_d)$ are unitary matrices having $d$ distinct eigenvalues $\omega^i$ with $i \in \{ 0, \cdots, d - 1 \}$.   
 
One important point to be stressed is that the measurements certified here are the optimal Collins-Gisin-Linden-Massar-Popescu (CGLMP) measurements \cite{SATWAP,cglmp1,cglmp2}. Thus these measurements have wide ranges of applications both in quantum information theory and cryptography ranging from witnessing the dimension of a Hilbert-space \cite{CGLMP_app1,CGLMP_app2,CGLMP_app3}, reducing quantum communication complexity \cite{CGLMP_app4,CGLMP_app5}, advantages in communication game \cite{CGLMP_app6}, remote preparation of quantum states \cite{CGLMP_app7}, distribution of secure key \cite{CGLMP_app8,CGLMP_app9} to generating genuine randomness \cite{CGLMP_app10}. Further, these measurements have already been realized experimentally  \cite{opticalsetup,CGLMP_exp}.
 
Also note that the measurements certified here have also been self-tested in \cite{saha} based on the quantum violation of the SATWAP Bell inequalities. However,  as mentioned earlier, this self-testing protocol  requires entanglement between two spatially separated particles (spatial separation is required here in order to avoid locality loophole in Bell violation). On the contrary, our certification protocol can be realized using temporal quantum correlation pertaining to a single particle. 

\subsection{Robustness analysis} \label{s5a}
So far, we have proposed a certification protocol for the $d$-outcome ideal measurement settings given in Theorem~\ref{theo} employing the maximal quantum violations of the temporal inequalities \eqref{temp} under Assumptions \ref{ass1}-\ref{ass2} when each of the measurements satisfies condition (\ref{lemma1e}). However, in a real experimental scenario there is always some unavoidable noise and hence, the ideal measurements are hardly realizable. Thus, the condition (\ref{lemma1e}) may not be satisfied and/or one may not get $\tau_d = 4(d-1)$ under Assumptions \ref{ass1}-\ref{ass2}. In such cases, we ask the question whether we can certify those measurements up to a certain threshold.  The term robustness of the certification protocol implies that the non-ideal measurements are close to the ideal ones if the correlations produced by the non-ideal measurements are close the  ideal correlations. In the following, we present the robustness analysis of our certification protocol for the following two cases: $\bullet$ when satisfying the condition (\ref{lemma1e})  is affected by the non-ideal observables, $\bullet$ when the magnitude of the temporal inequality (\ref{temp}) is affected by the non-ideal   observables, whereas satisfying the condition (\ref{lemma1e})  remains unaffected.

Suppose that in an experimental situation  under Assumptions \ref{ass1}-\ref{ass2}, instead of performing measurements of the ideal  unitary observables $A_i$ with $i \in \{1,2,3,4\}$   mentioned in the statements of Theorem \ref{theo}, measurements of the  non-ideal  observables $\widetilde{A}_1$, $\widetilde{A}_2$, $\widetilde{A}_3$, $\widetilde{A}_4$ are being performed, where each of these non-ideal observables does not satisfy the condition (\ref{lemma1e}).




Now, let us present the following theorem (for proof, see Appendix \ref{app4n}) that represents the robustness analysis for any potential imprecision in satisfying the condition (\ref{lemma1e}).

\begin{thm}
Suppose the non-ideal observables $\widetilde{A}_1$, $\widetilde{A}_2$, $\widetilde{A}_3$, $\widetilde{A}_4$ satisfy the following,
\begin{align}
p(a_i | \widetilde{A}_i) &=	p(a_i, a_i | \widetilde{A}_i, \widetilde{A}_i)  + \eta_i^{(a_i)} \hspace{0.1cm} \text{with} \hspace{0.2cm} \eta_i^{(a_i)}>0 \nonumber \\
	&\forall \, \, a_i \in \{0,  \cdots, d-1\}, \, \, \forall \, \, i \in \{1, 2, 3, 4\}.
	\label{lemma1er}
\end{align}
Then we have for all $a_i \in \{0,  \cdots, d-1\}$ and for all $i \in \{1, 2, 3, 4\}$
\begin{align}
	&\Bigg| \Bigg| \rho^{(\mathcal{P})} \left[ \widetilde{K}_i^{a_i^{\dagger}}  \widetilde{K}_i^{a_i}   -	 \Big( \widetilde{K}_i^{a_i^{\dagger}} \Big)^2 \Big( \widetilde{K}_i^{a_i}  \Big)^2 \right] \rho^{(\mathcal{P})} \nonumber \\
	& \hspace{0.9cm}  - \rho^{(\mathcal{P})} \left[ K_i^{a_i^{\dagger}}  K_i^{a_i}    -	 \Big( K_i^{a_i^{\dagger}} \Big)^2 \Big( K_i^{a_i}  \Big)^2  \right] \rho^{(\mathcal{P})} \Bigg| \Bigg|_{\text{HS}}  \nonumber \\
	& < \eta_i^{(a_i)},
	\label{lero7th}
\end{align}
where $\{\widetilde{K}_i^{a_i}\}$ and $\{K_i^{a_i}\}$ are the Kraus operators defined in \eqref{krausop} of the non-ideal observable $\widetilde{A}_i$ and the ideal observable $A_i$ respectively. Here,  $||\cdot||_{\text{HS}}$ denotes the Hilbert–Schmidt  norm. 
	\label{thm2n}
\end{thm}

In the above theorem, the condition $\eta_i^{a_i}>0$ for all $a_i \in \{0, \cdots, d-1\}$ and for all $i \in \{1, 2, 3, 4\}$ naturally appears due to the fact that $p(a_i, a_i | A_i, A_i) \leq p(a_i | A_i)$ for all $a_i \in \{0, \cdots, d-1\}$ and for all $i \in \{1, 2, 3, 4\}$ for any input state prepared by $\mathcal{P}$.

Next, let us consider that each of the four non-ideal observables $\widetilde{A}_1$, $\widetilde{A}_2$, $\widetilde{A}_3$, $\widetilde{A}_4$ satisfies the condition \eqref{lemma1e}. Hence, the measurement effects of each of these observables are mutually orthogonal projectors, which implies that these four non-ideal observables are unitary. Now, consider that   in the aforementioned scenario with any fixed value of $d$ the magnitude of  the temporal inequality (\ref{temp}) with these non-ideal unitary observables is $4(d-1)-\epsilon$, where $\epsilon$ is a positive  number.  Note that the maximum quantum violation of the temporal inequality (\ref{temp})  under Assumptions \ref{ass1}-\ref{ass2}  is $4(d-1)$ when each of the four  observables satisfies the condition \eqref{lemma1e}, thereby implying that $\epsilon$ cannot be negative.

Against the above backdrop, we present the following theorem (for proof, see Appendix \ref{app4}) that represents the robustness analyses of our certification scheme associated with the magnitude of quantum violation of the temporal inequality (\ref{temp}).

\begin{thm}
 If the quantum value of the temporal expression $\tau_d$ given in \eqref{temp} for any fixed $d$ realized by unknown unitary (Fourier
transformed) observables $\widetilde{A}_1$, $\widetilde{A}_2$, $\widetilde{A}_3$, $\widetilde{A}_4$ satisfying the condition (\ref{lemma1e})  is $[4(d-1)-\epsilon]$ with $\epsilon$ being  a positive  number, then the following relations hold true
\begin{align}
(i) \ &\Bigg| \Bigg|  \left[A_1 \big(a_1 A_3^{\dagger} + a_1^{*} \omega A_4^{\dagger}\big) \right] \rho^{(\mathcal{P})} \nonumber \\
& \hspace{0.4cm} -  \left[\widetilde{A}_1 \big(a_1 \widetilde{A}_3^{\dagger} + a_1^{*} \omega \widetilde{A}_4^{\dagger}\big) \right] \rho^{(\mathcal{P})}  \Bigg| \Bigg|_{\text{HS}}  < \sqrt{ \epsilon}.
\label{robustness1}
\end{align}
 \begin{align}
(ii)\ &\Bigg| \Bigg| \left[ A_2 \big(a^{*}_1 A_3^{\dagger} + a_1 A_4^{\dagger}\big) \right] \rho^{(\mathcal{P})} \nonumber \\
& \hspace{0.4cm} -  \left[ \widetilde{A}_2 \big(a^{*}_1 \widetilde{A}_3^{\dagger} + a_1 \widetilde{A}_4^{\dagger}\big) \right]  \rho^{(\mathcal{P})} \Bigg| \Bigg|_{\text{HS}}   < \sqrt{\epsilon}.
\label{robustness2}
\end{align}
\begin{align}
(iii) \ &\Bigg| \Bigg|  \left( A_4 A_3^{\dagger} - \omega A_3 A_4^{\dagger}\right)  \rho^{(\mathcal{P})} \nonumber \\
& \hspace{1.1cm} -  \left(\widetilde{A}_4 \widetilde{A}_3^{\dagger} - \omega \widetilde{A}_3 \widetilde{A}_4^{\dagger} \right) \rho^{(\mathcal{P})} \Bigg| \Bigg|_{\text{HS}}  \nonumber \\
&\leq 2  \sqrt{\epsilon} \big(2 + \sqrt{\epsilon} \big).
\label{robustness1new}
\end{align}
\begin{align}
(iv) \ &\Bigg| \Bigg|  \left(\omega A_2 A_1^{\dagger} -  A_1 A_2^{\dagger}\right)  \rho^{(\mathcal{P})} \nonumber \\
& \hspace{1.1cm}-  \left( \omega \widetilde{A}_2 \widetilde{A}_1^{\dagger} -  \widetilde{A}_1 \widetilde{A}_2^{\dagger}\right) \rho^{(\mathcal{P})} \Bigg| \Bigg|_{\text{HS}} \nonumber \\
&\leq 2 \sqrt{\epsilon} \big(2 + \sqrt{\epsilon} \big).
\label{robustness2new}
\end{align}
where $A_1,A_2,A_3,A_3$ are any set of unitary (Fourier
transformed) observables  that satisfies the condition (\ref{lemma1e}), achieves $\tau_d = 4(d-1)$ and thus satisfies Theorem \ref{theo}.
\label{thm2}
\end{thm}

For a more general robustness analysis, one should consider that both  the magnitude of the temporal inequality (\ref{temp}) and satisfying  the condition (\ref{lemma1e}) are affected simultaneously by the non-ideal measurements. We leave this question  for future study.


However, if the difference between the experimentally measured values of $p(a_i, a_i | \widetilde{A}_i, \widetilde{A}_i)$ and $p(a_i | \widetilde{A}_i)$ is within the statistical error range, then it can be approximated that the observables $\widetilde{A}_i$ ($i \in \{1,2,3,4\}$) satisfy the conditions (\ref{lemma1e}). In such cases, Theorem \ref{thm2} alone presents the complete robustness analysis of our certification protocol.
 
Another important point to be stressed here is that if an additional assumption is taken into account, then the certification of the measurements presented in Eq.(\ref{certifiedmeasurements}) can be demonstrated without requiring the measurements to satisfy the conditions (\ref{lemma1e}) (see the next Sec. \ref{s5b} for details). Hence, in this case, Theorem \ref{thm2n} is not required for demonstrating robustness of our protocol.

\subsection{Robust certification without using Lemma \ref{lemma1} or condition \eqref{lemma1e}} \label{s5b}

Now, we will show that our certification protocol can be formulated based on the quantum violation of the temporal inequalities (\ref{temp}) alone without using the condition mentioned in Lemma \ref{lemma1} if we consider another assumption as described below together with the Assumptions \ref{ass1} and \ref{ass2}. 



\begin{assumption}
The measurements $A_i \equiv \{M_i^{a_i} \}$ with $i \in \{1,2,3,4\}$ are realized in a particular way such that the Kraus operators $K_i^{a_i}$ are Hermitian or, equivalently, $K_i^{a_i} = \sqrt{M_i^{a_i}}$ for all $a_i$ and for all $i$, where $\{M_i^{a_i}\}$ are the measurement effects. In other words, any state updates due to a measurement following the the L\"{u}ders rule.
\label{ass3}
\end{assumption}
Note that the L\"{u}ders rule for state evolution due to quantum measurement appears in the context of unsharp measurements \cite{luders1} and other scenarios \cite{luders2,luders3}. Now, under Assumptions \ref{ass1}-\ref{ass3} the above-mentioned certification protocol using the temporal inequalities (\ref{temp}) alone is robust as stated below.


\begin{thm}\label{theo4}
Suppose that in the scenario considered by us with any fixed value  of $d$  under Assumptions \ref{ass1}-\ref{ass3}, the quantum violation of the temporal inequality (\ref{temp})  is $4(d-1)$, which is achieved by unknown measurements $A_1,A_2,A_3,A_4$ acting on some $\mathbb{C}^D$. Then, for any $d$, $\mathbb{C}^D= \mathbb{C}^d \otimes \mathcal{H}'$ and there exists a  unitary transformation $U$ such that Eq. \eqref{certifiedmeasurements} holds true. 

Moreover, under Assumptions \ref{ass1}-\ref{ass3}, if a quantum violation $4(d-1)-\epsilon$ is achieved by non-ideal measurements $\widetilde{A}_1,\widetilde{A}_2,\widetilde{A}_3,\widetilde{A}_4$ for any non-negative $\epsilon$, then the following two relations hold for all $ k \in \{1, \cdots, d-1\}$,
\begin{align} 
 \Bigg| \Bigg| & \left[A_1^{(k)}  \left( a_k  A_3^{(k)^\dagger} + a_k^{*} \omega^k A_4^{(k)^\dagger} \right) \right] \rho^{(\mathcal{P})} \nonumber \\
 & \hspace{0.65cm} -  \left[\widetilde{A}_1^{(k)} \left( a_k \widetilde{A}_3^{(k)^\dagger} + a_k^{*} \omega^k \widetilde{A}_4^{(k)^\dagger} \right) \right] \rho^{(\mathcal{P})}  \Bigg| \Bigg|_{\text{HS}}  \nonumber \\
 &< \sqrt{\epsilon} ,
 \label{robustnessnew1}
 \end{align}
 \begin{align} 
 \Bigg| \Bigg| & \left[A_2^{(k)}  \left( a_k^{*} A_3^{(k)^\dagger} + a_k  A_4^{(k)^\dagger} \right) \right] \rho^{(\mathcal{P})} \nonumber \\
 & \hspace{0.65cm} -  \left[\widetilde{A}_2^{(k)} \left( a_k^{*} \widetilde{A}_3^{(k)^\dagger} + a_k  \widetilde{A}_4^{(k)^\dagger} \right) \right] \rho^{(\mathcal{P})}  \Bigg| \Bigg|_{\text{HS}}  \nonumber \\
 &< \sqrt{\epsilon} .
 \label{robustnessnew2}
 \end{align}
\end{thm}

\begin{proof}
Let us first note that the expression of the joint probability $p(a_i,a_j|A_i,A_j)$ given by \eqref{jointsimple} immediately reduces to 
\begin{align}
p(a_i,a_j|&A_i,A_j) =\text{Tr} \Big[ M_j^{a_j} \,M_i^{a_i} \, \rho^{(\mathcal{P})} \Big] \nonumber \\
&\, \, \, \forall \, i, j \in \{1, 2, 3, 4\}, \nonumber \\
& \, \, \, \text{and} \, \, \,  \forall \, a_i, a_j \in \{0,  \cdots, d-1\} 
\label{pro1newnewnewnew}
\end{align} 
if each Kraus operator $K_i^{a_i}$ is taken to be $\sqrt{M_i^{a_i}}$. Consequently, using Eqs.(\ref{exp1}), (\ref{obsnew}) and (\ref{pro1newnewnewnew}), we have for all $k, l =1, \cdots, d-1$ and for all $i, j \in \{1, 2, 3, 4\}$,
\begin{align}
 \langle A_i^{(k)} \, A_j^{(l)}\rangle = \text{Tr}\Big[A_i^{(k)} \, A_j^{(l)} \, \rho^{(\mathcal{P})} \Big]. \nonumber
\end{align}
Therefore, in this case,  the left hand sides of the temporal inequalities (\ref{temp}) under Assumptions \ref{ass1}-\ref{ass3}  can be expressed in the form (\ref{temp2}) without requiring the condition (\ref{lemma1e}) to be satisfied by the four measurements.

The only difference here from the previous calculation is the fact that  $A_i^{(k)^\dagger} A_i^{(k)} \leq \mathbbm{1}$ for all $i$ and $k$. 
Following the exact steps done from Eq.\eqref{temp2} to Eq.\eqref{sos5}, one obtains
\begin{align}
    \sum_{k=1}^{d-1} \sum_{x=1}^{2} \Big[ \Big(P_x^{(k)}\Big)^{\dagger} \, \Big(P_x^{(k)}\Big) \Big] \leq  4 (d-1) \, \mathbbm{1} - \hat{\beta}_{\tau_d},
    \label{sosnewpovm}
\end{align} 
where $P_x^{(k)}$ is defined in \eqref{pxk} and the equality in (\ref{sosnewpovm}) holds only if $A_i^{(k)^\dagger} A_i^{(k)} = \mathbbm{1}$ for all $i$ and $k$. 
Since the left hand side of Eq. (\ref{sosnewpovm}) is sum of positive operators, we also have Eq. \eqref{maxqvalue}. Now, we know a quantum realization \eqref{certifiedmeasurements} that achieves $\tau_d=4(d-1)$ and at the same time satisfies Assumptions \ref{ass1}-\ref{ass3}. This implies that, even  without condition \eqref{lemma1e}, the maximum quantum magnitude of $\tau_d$ is $4(d-1)$  under Assumptions \ref{ass1}-\ref{ass3}.
More importantly, when $\text{Tr}\big[ \rho^{(\mathcal{P})} \, \hat{\beta}_{\tau_d} \big] = 4(d-1)$ is attained, we must have equality in Eq. \eqref{sosnewpovm}.
Consequently, the equality in the sum-of-square decomposition \eqref{sosnewpovm} implies $A_i^{(k)^\dagger} A_i^{(k)} = \mathbbm{1}$ for all $i$ and $k$.

Now, as shown in \cite{Kaniewski}, $A_i^{(k)^\dagger} A_i^{(k)} = \mathbbm{1}$ if and only if the measurement effects $ \{M_i^{a_i}\}$ are mutually orthogonal projectors. It is, therefore, implied that when the maximal quantum violation of the temporal inequality (\ref{temp}) with any fixed $d$ under the Assumptions \ref{ass1}-\ref{ass3} is attained, then the observables $A_i^{(k)}$ are unitary and $A_i^{(k)} = A_i^{k}$ for all $i$ and $k$. Therefore, in this case, the whole proof of Theorem \ref{theo} remains valid without invoking the condition (\ref{lemma1e}). In other words, one can certify the observables $A_i$ with $i \in \{1, 2, 3, 4\}$ only using the temporal inequalities (\ref{temp}) under Assumptions \ref{ass1}-\ref{ass3}.

For the robustness part, note that the relations (\ref{robustnessnew1}-\ref{robustnessnew2}) are similar to the previously derived robustness relations (\ref{robustness1}-\ref{robustness2}) in Theorem \ref{thm2}. One can verify that (\ref{robustnessnew1}-\ref{robustnessnew2}) can be derived following the exact steps used for deriving (\ref{robustness1}-\ref{robustness2}) taking all $k \in \{1,\cdots, d-1\}$. The only difference in the present case is the fact that $\widetilde{A}_i^{(k)^\dagger} \widetilde{A}_i^{(k)} \leq \mathbbm{1}$ for all $k \in \{1, \cdots, d-1\}$ instead of strict equality in case of each non-ideal observable $\widetilde{A}_i$ with $i \in \{1,2,3,4\}$. 
\end{proof}

\section{Secure randomness certification} \label{s6}
Here, we  present a protocol for the  secure certification  of randomness as a 
relevant application of our proposed formalism for the certification of $d$-outcome quantum measurements. In particular, let us consider that the Assumptions \ref{ass1} and \ref{ass2} are satisfied    in the scenario considered by us with any fixed value of $d$. Further, we also assume that the condition (\ref{lemma1e}) is satisfied by the each of the four  unitary observables $A_1$, $A_2$, $A_3$, $A_4$ and the magnitude of the temporal inequality (\ref{temp}) is $4(d-1)$ using the above four measurements.

Consider a scenario, where a party, say, Eve  prepares the initial state. In other words, the internal functioning of the preparation device $\mathcal{P}$ is controlled by Eve. In each experimental run, Eve prepares a pure state $|\psi^{(\mathcal{P})}_x\rangle$ in such a way that, on average, the initial state becomes  $\rho^{(\mathcal{P})} = \mathbbm{1}/D$. Let us assume that Eve knows beforehand which two measurements will be performed sequentially in each run. In such a scenario, Eve can always predict the outcome of the first measurement. For example, consider an experimental run in which measurement of $A_i$ is performed at first on the preparation, and then the measurement of $A_j$ is performed (where $i, j \in \{1,2,3,4\}$). In this case, Eve can predict the outcome of the first measurement by preparing an eigenstate of $A_i$.  Hence, no randomness can be   certified securely from the first measurement and, 
therefore, we will focus on randomness certification using the outcome statistics of the second measurement.

Since, we are not interested in the randomness certification using the outcome statistics of the first measurement, it excludes certifying classical randomness associated with the preparation $\rho^{(\mathcal{P})} = \mathbbm{1}/D$. In other words, although the preparation device prepares a maximally mixed state, the randomness from the maximally mixed state is not genuine quantum randomness,  rather this randomness is a manifestation of the classical convex mixture of different pure states. However, after the first measurement, the state collapses to a different pure state, and thus the outcome of the second measurement provides genuine quantum randomness. This is why we will certify randomness from  $p(a_j|A_i,A_j,a_i) = p(a_i, a_j|A_i, A_j)/p(a_i|A_i)$.

Let us now define the measure of randomness, $\mathcal{H}(A_i,A_j)$ for a fixed set of two  observables $\lbrace A_i, A_j \rbrace$ as,
\begin{align}
    \mathcal{H}(A_i,A_j) & =  \min_\mathcal{S} \left[ - \sum_{a_i=0}^{d-1} p(a_i|A_i) \sum_{a_j=0}^{d-1} \alpha \log_2 \alpha \right], \nonumber \\
    &\text{with} \, \, \alpha = p(a_j|A_i,A_j,a_i)
    \label{ranquan}
\end{align}
and $\mathcal{S}$ denoting all possible strategies of Eve for preparing $A_i$ reproducing the observed probabilities.
The above quantification is based on the Shannon entropy that characterizes the average randomness involved in the probability distributions \cite{shannon}. Further, we have taken average of it over all possible outcomes of the first measurement. 


Moreover, to quantify the genuine or guaranteed randomness we have to consider the minimum in (\ref{ranquan}) over all possible Eve's strategy of preparing the four  observables $A_i$ that satisfy the condition (\ref{lemma1e}) and gives $\tau_d = 4(d-1)$ under Assumptions \ref{ass1}-\ref{ass2}.

Since, the condition (\ref{lemma1e}) is satisfied by the each of the four unknown  observables $A_1$, $A_2$, $A_3$, $A_4$ and the magnitude of the temporal inequality (\ref{temp})  is $4(d-1)$, Theorem \ref{theo} implies that there exists $U$ such that
\begin{align}
U \Pi_1^{a_1} U^{\dagger} & = \widetilde{\Pi}_1^{a_1} \otimes  \mathbbm{1}_{\mathcal{H}'} = |Z^{a_1}_d\rangle \langle Z^{a_1}_d| \otimes  \mathbbm{1}_{\mathcal{H}'},  \nonumber \\
U \Pi_2^{a_2} U^{\dagger} &= \widetilde{\Pi}_2^{a_2} \otimes  \mathbbm{1}_{\mathcal{H}'}= |T^{a_2}_d\rangle \langle T^{a_2}_d| \otimes  \mathbbm{1}_{\mathcal{H}'},  \nonumber \\
U \Pi_3^{a_3} U^{\dagger} &= \widetilde{\Pi}_3^{a_3} \otimes  \mathbbm{1}_{\mathcal{H}'}= |M^{a_3}_d\rangle \langle M^{a_3}_d| \otimes  \mathbbm{1}_{\mathcal{H}'},  \nonumber \\
U \Pi_4^{a_4} U^{\dagger} &= \widetilde{\Pi}_4^{a_4} \otimes  \mathbbm{1}_{\mathcal{H}'}= |N^{a_4}_d\rangle \langle N^{a_4}_d| \otimes  \mathbbm{1}_{\mathcal{H}'}, \nonumber \\
&\hspace{1cm} \forall \, \, a_1, a_2, a_3, a_4 \in \{0, 1, \cdots, d-1\},
\label{projectorsran}
\end{align}
where $\lbrace \Pi_i^{a_i} \rbrace$ are the set of  mutually orthogonal projectors for measurement of $A_i$; $|Z^{a_1}_d\rangle$ is the eigenstate of $Z_d$  with eigenvalue $\omega^{a_1}$, $|T^{a_2}_d\rangle$ is the eigenstate of $T_d$  with eigenvalue $\omega^{a_2}$, $|M^{a_3}_d\rangle$ is the eigenstate of $M_d = a_1^{*} Z_d + 2 (a_1^*)^3 T_d$  with eigenvalue $\omega^{a_3}$, $|N^{a_4}_d\rangle$ is the eigenstate of $N_d = a_1 Z_d - a_1^* T_d $  with eigenvalue $\omega^{a_4}$. 
 Note that satisfying the condition (\ref{lemma1e})  by the each of the four observables $A_1$, $A_2$, $A_3$, $A_4$ implies that the measurement effects of each of these four observables are mutually orthogonal projectors. Moreover, if the magnitude of the temporal inequality (\ref{temp}) is achieved to be $4(d-1)$ using the above four observables, then Theorem \ref{theo} implies that these projectors are rank-one. This follows from the fact that each of these four measurements has $d$ number of possible outcomes and the dimension of each of the operators $Z_d$, $T_d$, $M_d$, $N_d$ is $d$. Thus,  we have taken each of the projectors in the above Eq.(\ref{projectorsran}) to be rank-one.


Next, we  evaluate the expression of $p(a_j|A_i,A_j,a_i)$ in order to find out  $\mathcal{H}(A_i,A_j)$ for a given $\lbrace A_i, A_j \rbrace$. From Eq.(\ref{projectorsran}), we can write the following,
\begin{widetext}
\begin{align}
    p(a_j|A_i,A_j,a_i)     &= \frac{\text{Tr} \Bigg[ \Big( \widetilde{\Pi}_i^{a_i} \otimes \mathbbm{1}_{\mathcal{H}'} \Big) \Big(U \rho^{(\mathcal{P})} U^{\dagger} \Big) \Big(\widetilde{\Pi}_i^{a_i} \otimes \mathbbm{1}_{\mathcal{H}'} \Big) \Big(\widetilde{\Pi}_j^{a_j} \otimes \mathbbm{1}_{\mathcal{H}'} \Big) \Bigg]}{\text{Tr} \Bigg[ \Big(\widetilde{\Pi}_i^{a_i} \otimes \mathbbm{1}_{\mathcal{H}'} \Big) \Big(U \rho^{(\mathcal{P})} U^{\dagger} \Big) \Bigg]}.
    \nonumber 
\end{align}
\end{widetext}
Since $\rho^{(\mathcal{P})} = \mathbbm{1}/D$, the above expression can be simplified as $ p(a_j|A_i,A_j,a_i) = \frac{\text{Tr} \Big[ \widetilde{\Pi}_i^{a_i} \widetilde{\Pi}_j^{a_j} \widetilde{\Pi}_i^{a_i} \otimes \mathbbm{1}_{\mathcal{H}'}  \Big]}{\text{Tr}\Big[ \widetilde{\Pi}_i^{a_i}  \otimes \mathbbm{1}_{\mathcal{H}'}  \Big]} = \text{Tr} \Big[ \widetilde{\Pi}_i^{a_i} \widetilde{\Pi}_j^{a_j} \Big] = \Big| \langle A_i^{a_i} | A_j^{a_j} \rangle \Big|^2$, where $| A_i^{a_i} \rangle$ ($| A_j^{a_j} \rangle$) is the eigenstate of $A_i$ ($A_j$) with eigenvalue $\omega^{a_i}$ ($\omega^{a_j}$). In the above, we have used the fact that $\widetilde{\Pi}_i^{a_i}$ is a rank-one projector, implying that $ \text{Tr}\Big[ \widetilde{\Pi}_i^{a_i}   \Big] =1$. 

On the other hand, we have for $\rho^{(\mathcal{P})} = \mathbbm{1}/D$
\begin{align}
    p(a_i|A_i) &= \text{Tr} \Bigg[ \Big(\widetilde{\Pi}_i^{a_i} \otimes \mathbbm{1}_{\mathcal{H}'} \Big) \Big(U \rho^{(\mathcal{P})} U^{\dagger} \Big) \Bigg] \nonumber \\
    & = \frac{1}{D} \text{Tr} \Bigg[ \Big(\widetilde{\Pi}_i^{a_i} \otimes \mathbbm{1}_{\mathcal{H}'} \Big) \Bigg] \nonumber \\
    & = \frac{1}{d},
    \nonumber 
\end{align}
where we have used $ \text{Tr}\Big[ \widetilde{\Pi}_i^{a_i}   \Big] =1$ and $ \text{Tr}\Big[ \mathbbm{1}_{\mathcal{H}'}  \Big] =D/d$.
Therefore, Eq.(\ref{ranquan}) becomes
\begin{widetext}
\begin{align}
     \mathcal{H}(A_i,A_j) = \min_\mathcal{S} \left[ -  \sum_{a_j=0}^{d-1} \Big| \langle A_i^{a_i} | A_j^{a_j} \rangle \Big|^2 \log_2 \Big| \langle A_i^{a_i} | A_j^{a_j} \rangle \Big|^2 \right].
    \label{ranquan21}
\end{align}
\end{widetext}
As mentioned earlier, $\mathcal{S}$ denotes all possible strategies of Eve for preparing $A_i$ reproducing the observed probabilities. All these strategies are connected unitarily as can be seen from Eq.(\ref{projectorsran}). However, the expression  $\sum_{a_j=0}^{d-1} \Big| \langle A_i^{a_i} | A_j^{a_j} \rangle \Big|^2 \log_2 \Big| \langle A_i^{a_i} | A_j^{a_j} \rangle \Big|^2$ is independent of the unitary $U$. Hence, from Eq.(\ref{ranquan21}), we have for the present case
\begin{align}
     \mathcal{H}(A_i,A_j) =  -  \sum_{a_j=0}^{d-1} \Big| \langle A_i^{a_i} | A_j^{a_j} \rangle \Big|^2 \log_2 \Big| \langle A_i^{a_i} | A_j^{a_j} \rangle \Big|^2 .
    \label{ranquan2}
\end{align}

Here it should be mentioned that if $i=j$, i.e., if $A_i=A_j$, then from the above Eq.(\ref{ranquan2}), we have $\mathcal{H}(A_i,A_j) = 0$. Hence, no randomness can be certified.

Next, consider that $i \neq j$. At first, let us take $i=1$ and $j=2$. In other words, we are considering the case when the measurement of $A_1$ is performed at first on the initial preparation and then the measurement of  $A_2$ is performed. For this case, Eq. (\ref{ranquan2}) reduces to
\begin{align}
    \mathcal{H}(A_1,A_2) = -  \sum_{a_2=0}^{d-1} \Big| \langle Z_d^{a_1} | T_d^{a_2} \rangle \Big|^2 \log_2 \Big| \langle Z_d^{a_1} | T_d^{a_2} \rangle \Big|^2.
    \label{ranquan3}
\end{align}
This can be evaluated using the expressions of the eigenstates of $Z_d$ and $T_d$. The spectral decomposition of $Z_d$ and $T_d$ are given by \cite{saha},
\begin{align}
    Z_d = \sum_{q=0}^{d-1}\omega^q |q\rangle \langle q|, ~~\text{and}~~  T_d =  \sum_{r=0}^{d-1}\omega^r |r\rangle \langle r|_{T_d}
    \nonumber
\end{align}
with
\begin{align}
  |r\rangle_{T_d} = \frac{2}{d} \sum_{q=0}^{d-1} (-1)^{\delta_{q,0}} \frac{\omega^{-\frac{q}{2}}}{1- \omega^{r-q -\frac{1}{2}}} |q\rangle
  \nonumber
\end{align}
Using these, we have
\begin{align}
   \Big| \langle Z_d^{a_1} | T_d^{a_2} \rangle \Big|^2= \frac{4}{d^2} \frac{1}{ \left \lvert 1-\omega^{a_1-a_2-\frac{1}{2}} \right \rvert^2} .
   \label{overlap}
\end{align}
 Now, replacing the variable $(a_1-a_2)$ by $x$ and using the fact that $\omega^d=1$, we get
\begin{widetext}
\begin{align}
     \mathcal{H}(A_1,A_2) 
     &= - \sum_{x=0}^{d-1} \frac{4}{d^2} \frac{1}{ \left \lvert 1-\omega^{x-\frac{1}{2}} \right \rvert^2} \log_2 \frac{4}{d^2} \frac{1}{ \left \lvert 1-\omega^{x-\frac{1}{2}} \right \rvert^2}.
    \label{ranquan41}
\end{align}
\end{widetext}

Next, let us take $i=2$ and $j=1$. In other words, we are considering the case when the measurement of $A_2$ is performed at first on the initial preparation and then the measurement of $A_1$ is performed. In this case, one can easily check that 
\begin{align}
     \mathcal{H}(A_2,A_1) &= \mathcal{H}(A_1,A_2) \nonumber 
    \label{ranquan42}
\end{align}

We now calculate $\mathcal{H}(A_1,A_2)$ for several values of $d$ and the variation is plotted in Fig. \ref{fig2}. From this figure, it is evident that $\mathcal{H}(A_1,A_2)$ increases with $d$. 

 \begin{figure}[t]
    \centering
    \includegraphics[scale=0.3]{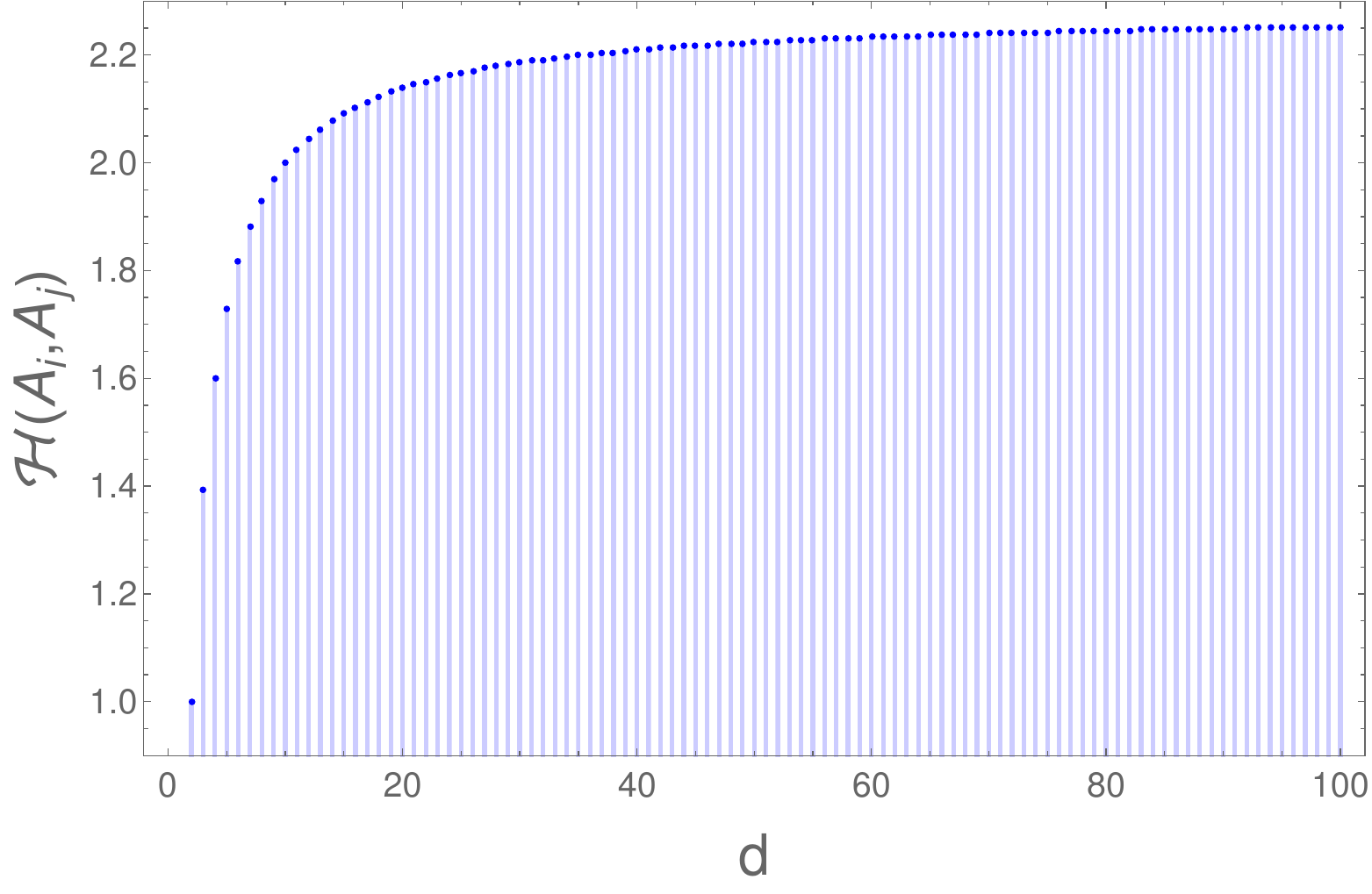}
    \caption{The blue dotted lines demonstrate the variation of the amount of output  randomness $\mathcal{H}(A_i,A_j)$ produced in the present certification protocol with the number of outcomes  $d$. Here, $\mathcal{H}(A_i,A_j)$ is the amount of randomness produced from the second measurement when the measurement of $A_i$ is performed at first on the initial preparation and then the measurement of $A_j$ is performed. This plot is for four possible cases- (1) when $i=1$ and $j=2$, (2) when $i=2$ and $j=1$, (3) when $i=3$ and $j=4$, (4)  when $i=4$ and $j=3$.}
    \label{fig2}
\end{figure}

Since $M_d = a_1^{*} Z_d + 2 (a_1^*)^3 T_d$  and $N_d = a_1 Z_d - a_1^* T_d $ are connected to $Z_d$ and $T_d$, respectively, with the same unitary $W$ (see
proof of Theorem \ref{theo}), we have $\Big| \langle A_1^{u} | A_2^{v} \rangle \Big|^2 = \Big| \langle A_3^{u} | A_4^{v} \rangle \Big|^2$ for all $u, v \in \{0,1, \cdots, d-1\}. $. Hence, for the case with $i=3$, $j=4$ and for the case with $i=4$, $j=3$, we have 
\begin{align}
   \mathcal{H}(A_3,A_4) = \mathcal{H}(A_4,A_3)  = \mathcal{H}(A_1,A_2).
\end{align}
 
For other combinations of $i \neq j \in \{1,2,3,4\}$, we have not evaluated $\mathcal{H}(A_i,A_j)$ in the present study, and we leave it as an open question.

\section{Conclusions} \label{s7}
Formulating  efficient certification protocols for quantum measurements  requiring  fewer assumptions or trusts on the preparation device is a worthwhile enterprise that will be helpful for establishing secure quantum information theoretic and cryptographic applications. In this work we have proposed a novel framework for certification of  a particular set of $d$-outcome quantum measurements, which does not require entanglement or any other spatial quantum correlation. Further, our protocol does not need any prior knowledge or assumption about the dimension of the system on which the measurements are performed.  Importantly, the specific measurements certified in the present study has fundamental significance as well as information theoretic applications. Considering a scenario consisting  of a preparation followed by two measurements in sequence,  we have first proposed a class of inequalities involving temporal correlations, and have then established their sum-of-squares decompositions. Using quantum violations of these inequalities, we have certified   a specific set of $d$-outcome quantum measurements uniquely up to some unitary freedom. Moreover, we have shown that our certification protocol is robust against non-ideal realizations. As a proposed application, our protocol can be used to generate genuine quantum randomness.

It needs to be emphasized that  one cannot certify the measurements uniquely without any assumption whatsoever on the preparation device, employing the quantum violations of our proposed temporal inequalities.   Starting with a preparation device producing a maximally mixed states of dimension $D$, one can certify some particular $d$-outcome quantum measurements of dimension $D$ with $D \geq d$ following our protocol. Further, no information about $D$ is required for realizing this protocol. Therefore, in our certification scheme, preparing $D$ number of mutually orthogonal pure states in $\mathbb{C}^D$ in different experimental runs randomly by the preparation device is sufficient.  In comparison, tomography of such measurements requires $\mathcal{O}(D^2)$ characterized quantum preparations \cite{Gianani20}. Hence,  the requirement on the preparation device in our protocol is less demanding than that in the case of tomography of quantum measurements. Additionally, tomography also requires some prior knowledge about measurement devices, in contrast to those in our scheme that behave essentially as black boxes without memory.

 Unlike certification protocols of $d$-outcome measurements proposed earlier \cite{saha,ghzsarkar}, entanglement between two spatially separated particles is not necessary for the successful realization of our protocol. However, in order to realize our protocol, one must trust that the preparation device produces maximally mixed state of a single particle. Since, for a given dimension, preparing a particular mixed state 
	of a single particle  is easier than preparing two spatially separated entangled particles, our protocol is less demanding to be verified experimentally, and should be more desirable for commercial purposes.

Before concluding, it would be pertinent to mention that the analysis presented in this article for certifying quantum components received from an unknown provider also furnishes a general methodology for introducing new  Leggett-Garg type  temporal inequalities following the structure of the existing Bell inequalities. Adopting this methodology, certification schemes for quantum devices using temporal correlations without requiring entanglement can be designed based on the existing certification protocols in the Bell scenarios involving entanglement.  For example, based on the particular Bell-type inequalities proposed in \cite{ghzsarkar} applicable in the $N-m-d$ scenario (involving $N$ parties, $m$ measurement settings per party, $d$ outcomes per measurement setting) with $N, m, d$ being arbitrary, one can propose Leggett-Garg type temporal inequalities involving $m$ number of $d$-outcome measurements adopting the methodology described here. Further, one can use our method to propose self-testing/certification protocols of $m$ number of $d$-outcome quantum measurements using temporal correlations without using entanglement between spatially separated particles based on the self-testing proof derived in  \cite{ghzsarkar}. Also, the method presented here can be further applied in the context of self-testing proof proposed in \cite{mub} to devise certification scheme of three-outcome mutually unbiased quantum measurements using temporal correlations that may be useful for secure certification of larger amount of randomness.

 To summarize, though the present study proposes certification protocol of some specific quantum measurements with arbitrary number of outcomes, the method presented here is quite general and can be immediately applied in different contexts to certify a wide range of quantum measurements employing temporal quantum correlations. Finally, it is worth noting that the certified $d$-outcome quantum measurements can be rigorously implemented in optical setups \cite{opticalsetup}.  


\subsection*{Acknowledgements} 
DD and DS acknowledge Science and Engineering Research Board (SERB), Government of India for financial support through the National Post Doctoral Fellowship (File Nos.: PDF/2020/001358 and PDF/2020/001682). During the later phase of this work, the research of DD is supported by the Royal Society (United Kingdom) through the Newton International Fellowship (NIF $\backslash R1 \backslash 212007$). ASM acknowledges support from
the project no. DST/ICPS/QuEST/2018/79 of the Department of Science and Technology, Government of India.

\newpage

\onecolumngrid
\appendix
\section{Proof of Lemma \ref{lemma1}}\label{app1}

Consider that the following condition holds under Assumptions \ref{ass1}-\ref{ass2},
\begin{equation}
p(a_i, a_i | A_i, A_i) = p(a_i | A_i).
\label{lemma1e1}
\end{equation}
Also, since condition \eqref{conjtprob} holds for all $|\psi\rangle \in \mathbb{C}^D$ (where $D$ is arbitrary), we have 
\begin{align}
   p(a_i, a_i | A_i, A_i) \leq p(a_i | A_i) \hspace{0.3cm} \forall  |\psi\rangle \in \mathbb{C}^D.
   \label{lemma1e2}
\end{align}
Hence, the condition (\ref{lemma1e1}) for $\rho^{(\mathcal{P})} = \mathbbm{1}/D$ together with (\ref{lemma1e2}) implies that 
\begin{equation}
p(a_i, a_i | A_i, A_i) = p(a_i | A_i) \hspace{0.3cm} \forall  |\psi\rangle \in \mathbb{C}^D.
   \label{lemma1e3}
\end{equation}
From the above, it follows that
\begin{align}
    \langle \psi | \left(\sqrt{M_i^{a_i}}\right)^{\dagger} \left(U_i^{a_i}\right)^{\dagger} M_i^{a_i} U_i^{a_i} \sqrt{M_i^{a_i}} | \psi \rangle = \langle \psi | M_i^{a_i} | \psi \rangle \hspace{0.15cm} \forall  |\psi\rangle \in \mathbb{C}^D ,
    \label{lemma1e4}
\end{align}
which implies
\begin{align}
   \left( \sqrt{M_i^{a_i}} \right)^{\dagger} \left(U_i^{a_i}\right)^{\dagger} M_i^{a_i} U_i^{a_i} \sqrt{M_i^{a_i}}  =  M_i^{a_i} .
    \label{lemma1e5}
\end{align}
Let $M_i^{a_i} = \sum_{x=0}^{m-1} \lambda_{u} |\psi_{u} \rangle \langle \psi_{u} |$ with $0 < \lambda_{u} \leq 1$ for all $u \in \{0,  \cdots, m-1\}$; $\{ |\psi_{0}\rangle, \cdots, |\psi_{D-1}\rangle \}$ being an orthonormal basis in $\mathbb{C}^D$ and  $m$ ($1 \leq m \leq D$) is the rank of $M^{a_i}_i$. Putting this in Eq.(\ref{lemma1e5}), we get the following,
\begin{align}
 \Bigg(\sum_{u=0}^{m-1} \sqrt{\lambda_{u}} |\psi_{u}\rangle \langle \psi_{u} | \Bigg) \left(U_i^{a_i}\right)^{\dagger} \Bigg(\sum_{u=0}^{m-1} \lambda_{u} |\psi_{u} \rangle \langle \psi_{u} | \Bigg) U_i^{a_i} \Bigg(\sum_{u=0}^{m-1} \sqrt{\lambda_{u}} |\psi_{u} \rangle \langle \psi_{u} | \Bigg)  =  \sum_{u=0}^{m-1} \lambda_{u} |\psi_{u} \rangle \langle \psi_{u} |.
 \label{lemma1e6}
\end{align}
Let $(U_i^{a_i})^{\dagger} |\psi_{u} \rangle = |\phi_{u} \rangle$ for all $u \in \{0, \cdots, D-1\}$, where $\{ |\phi_{0}\rangle, \cdots, |\phi_{D-1}\rangle \}$ is another orthonormal basis in $\mathbb{C}^D$. Using this, we get from Eq.(\ref{lemma1e6}),
\begin{align}
   \Bigg(\sum_{u=0}^{m-1} \sqrt{\lambda_{u}} |\psi_{u}\rangle \langle \psi_{u} | \Bigg)  \Bigg(\sum_{u=0}^{m-1} \lambda_{u} |\phi_{u} \rangle \langle \phi_{u} | \Bigg)  \Bigg(\sum_{u=0}^{m-1} \sqrt{\lambda_{u}} |\psi_{u} \rangle \langle \psi_{u} | \Bigg)  =  \sum_{u=0}^{m-1} \lambda_{u} |\psi_{u} \rangle \langle \psi_{u} |.
 \label{lemma1e7}
\end{align}
Let $|\psi_{k} \rangle$ belongs to the orthonormal basis $\{ |\psi_{0}\rangle, \cdots, |\psi_{D-1}\rangle \}$ and $k \in \{0, \cdots, m-1\}$. With this, we get the following from  Eq.(\ref{lemma1e7}),
\begin{align}
 &\langle \psi_{k} | \Bigg[  \Bigg(\sum_{u=0}^{m-1} \sqrt{\lambda_{u}} |\psi_{u}\rangle \langle \psi_{u} | \Bigg)  \Bigg(\sum_{u=0}^{m-1} \lambda_{u} |\phi_{u} \rangle \langle \phi_{u} | \Bigg)  \Bigg(\sum_{u=0}^{m-1} \sqrt{\lambda_{u}} |\psi_{u} \rangle \langle \psi_{u} | \Bigg) \Bigg] |\psi_{k} \rangle  =  \langle \psi_{k} |  \Big(\sum_{u=0}^{m-1} \lambda_{u} |\psi_{u} \rangle \langle \psi_{u} | \Big) |\psi_{k} \rangle .
 \label{lemma1e8}
\end{align}
After simplifying, we get the following condition from Eq.(\ref{lemma1e8}),
\begin{align}
    \sum_{u=0}^{m-1} \lambda_{u} \Big| \langle \phi_{u} | \psi_{k} \rangle \Big|^2 = 1.
    \label{lemma1e9}
\end{align}
Since $\{ |\phi_{0}\rangle, \cdots, |\phi_{D-1}\rangle \}$ is an orthonormal basis in $\mathbb{C}^D$, we have 
\begin{align}
    \sum_{u=0}^{D-1}  \Big| \langle \phi_{u} | \psi_{k} \rangle \Big|^2  = 1.
    \label{lemma1e10}
\end{align}
Now, subtracting Eq.(\ref{lemma1e9}) from Eq.(\ref{lemma1e10}), we get the following,
\begin{align}
    \sum_{u=0}^{m-1} (1-\lambda_{u}) \Big| \langle \phi_{u} | \psi_{k} \rangle \Big|^2 + \sum_{u=m}^{D-1} \Big| \langle \phi_{u} | \psi_{k} \rangle \Big|^2 = 0.
    \label{lemma1e11}
\end{align}
The left hand side of the above equation is the sum of positive terms. This sum is zero if and only if each  term is zero. Hence, we have
\begin{align}
\langle \phi_{u} | \psi_{k} \rangle = 0 \hspace{0.3cm} \forall u \in \{m,  \cdots, D-1\},
\label{lemma1e12}
\end{align}
where the above holds for all $k \in \{0, \cdots, m-1\}$.

On the other hand, it can be shown that for each $u \in \{0, \cdots, m-1\}$, there exists at least one $| \psi_{k} \rangle$ with $k \in \{0, \cdots, m-1\}$, such that $\Big| \langle \phi_{u} | \psi_{k} \rangle \Big|^2 \neq  0$. The negation of this leads to a contradiction as follows. Suppose, there exists one particular $u \in \{0,  \cdots, m-1\}$, denoted by $\widetilde{u}$, such that $\Big| \langle \phi_{\widetilde{u}} | \psi_{k} \rangle \Big|^2 =  0$ for all $k \in \{0, \cdots, m-1\}$. Hence, it is implied that $|\phi_{\widetilde{u}} \rangle$ is mutually orthogonal to $| \psi_{k} \rangle$ for all $k \in \{0,  \cdots, m-1\}$. 
On the other hand, by definition, $|\phi_{\widetilde{u}} \rangle$ is mutually orthogonal to $|\phi_{u} \rangle$ for all $u \in \{m,  \cdots, D-1\}$.
Since, Eq.(\ref{lemma1e12}) holds for all $k \in \{0, \cdots, m-1\}$, we can construct the following set: $\left\{ | \psi_{0} \rangle, \cdots, | \psi_{m-1} \rangle, |\phi_{\widetilde{u}} \rangle, |\phi_{m} \rangle, \cdots, |\phi_{D-1} \rangle \right\}$ consisting of $(D+1)$ number of mutually orthogonal vectors. However, in a Hilbert space of dimension $D$, one cannot have more than $D$ number of mutually orthogonal vectors. 
Therefore, for each $u \in \{0, 1, \cdots, m-1\}$, there exists at least one $| \psi_{k} \rangle$ with $k \in \{0,  \cdots, m-1\}$, such that $\Big| \langle \phi_{u} | \psi_{k} \rangle \Big|^2 \neq  0$. 

Since Eq.(\ref{lemma1e11}) can be derived for all values of $k \in \{0, \cdots, m-1\}$, we know that there exists at least one $k$ for each $u \in \{0, \cdots, m-1\}$ such that
\begin{align}
     (1-\lambda_{u}) \Big| \langle \phi_{u}| \psi_{k} \rangle \Big|^2  = 0 \hspace{0.5cm} \text{and} \hspace{0.5cm} \Big| \langle \phi_{u}| \psi_{k} \rangle \Big|^2 \neq  0. 
    \label{lemma1e13}
\end{align}
In other words, we have the following,
\begin{align}
    \lambda_{u} = 1 \hspace{0.3cm} \forall u \in \{0, \cdots, m-1\}.
\end{align}
Hence, we have $M_i^{a_i} =  \sum_{u=0}^{m-1}  |\psi_{u} \rangle \langle \psi_{u} |$. That is, $M_i^{a_i}$ is a projector.  Now, if condition (\ref{lemma1e1}) holds for all $a_i \in \{0,  \cdots, d-1\}$, then each of the POVM elements $\{M_i^{a_i}\}$ is projector, i.e., $\left(M_i^{a_i}\right)^2 = M_i^{a_i}$ for all $a_i \in \{0,  \cdots, d-1\}$. Now, it can easily be shown that the  projectors $\{M_i^{a_i}\}$ are mutually orthogonal. 
Multiplying $M_i^{\widetilde{a}_i}$ (with $\widetilde{a}_i \in \{0, \cdots, d-1\}$) on the both sides of the normalization condition $\sum_{a_i=0}^{d-1} M_i^{a_i} = \mathbbm{1}$, we get $\left(M_i^{\widetilde{a}_i}\right)^2 + \sum_{\substack{a_i=0 \\ a_i \neq \widetilde{a}_i}}^{d-1} M_i^{a_i} M_i^{\widetilde{a}_i}  = M_i^{\widetilde{a}_i}$. Since, as mentioned earlier, $\left(M_i^{\widetilde{a}_i}\right)^2 = M_i^{\widetilde{a}_i}$, we have the following, 
\begin{equation}
  \sum_{\substack{a_i=0 \\ a_i \neq \widetilde{a}_i}}^{d-1} M_i^{a_i} M_i^{\widetilde{a}_i} = 0.
  \label{orthoapp}
\end{equation}
The left hand side of this condition (\ref{orthoapp}) is a sum of products of two projectors. Now, product of two projectors is a positive operator \cite{product}.  Hence, the left hand side of  (\ref{orthoapp}) is a sum of positive operators. Therefore, we have that $M_i^{a_i} M_i^{\widetilde{a}_i} = 0$ for all $a_i \neq  \widetilde{a}_i \in \{0, \cdots, d-1\}$.

Now, note that the condition (\ref{lemma1e12}) is satisfied for all $k \in \{0, \cdots, m-1\}$ by the two orthonormal basis $\{ |\psi_{0}\rangle, \cdots, |\psi_{D-1}\rangle \}$ and $\{ |\phi_{0}\rangle, \cdots, |\phi_{D-1}\rangle \}$ in $\mathbb{C}^D$. Hence, it is implied from  (\ref{lemma1e12}) that the set $\left\{ | \psi_{0} \rangle, \cdots, | \psi_{m-1} \rangle,  |\phi_{m} \rangle, \cdots, |\phi_{D-1} \rangle \right\}$ is another orthonormal basis in $\mathbb{C}^D$. Therefore,  the $m$-dimensional subspace spanned by the vectors $\{ |\psi_{0}\rangle, \cdots, |\psi_{m-1} \rangle \}$ and the $m$-dimensional subspace  spanned by the vectors $\{ |\phi_{0}\rangle, \cdots, |\phi_{m-1} \rangle \}$ are the same. It is thus implied that 
\begin{align}
    \sum_{u=0}^{m-1}  |\psi_{u} \rangle \langle \psi_{u} | = \sum_{u=0}^{m-1}  |\phi_{u} \rangle \langle \phi_{u} | = \widetilde{\mathbbm{1}},
    \label{post2appnew}
\end{align}
where $\widetilde{\mathbbm{1}}$ is the identity operator acting on the $m$-dimensional subspace spanned by the vectors $\{ |\psi_{0}\rangle, \cdots, |\psi_{m-1} \rangle \}$. Since the above analysis holds for all possible choices of $U_i^{a_i}$, we have \eqref{UMU}.

\section{Analysis of the temporal inequality \eqref{temp} and its classical bound}\label{app2}

Consider the following quantity $\widetilde{\tau}_d$ which is a function of several probability distributions as introduced in \cite{SATWAP},
\begin{align}
    \widetilde{\tau}_d := \sum_{k=0}^{\floor*{d/2} -1 } \big[\alpha_k (\mathbb{P}^1_k + \mathbb{P}^2_k) - \beta_k (\mathbb{Q}^1_k +\mathbb{Q}^2_k) \big],
\end{align}
where the expressions $\mathbb{P}^1_k$, $\mathbb{P}^2_k$, $\mathbb{Q}^1_k$ and $\mathbb{Q}^2_k$ are defined as
\begin{equation}
   \mathbb{P}^1_k = p(A_1 = A_3 +k) + p(A_2 = A_3 - k) + p(A_2 = A_4 +k) + p(A_1=A_4-k-1),
\end{equation}
\begin{equation}
    \mathbb{Q}^1_k = p(A_1 = A_3 - k - 1 ) + p(A_2 = A_3 + k + 1) + p(A_2 = A_4 -k-1) + p(A_1 = A_4 + k),
\end{equation}
\begin{equation}
   \mathbb{P}^2_k = p(A_3 = A_1 +k) + p(A_3 = A_2 - k) + p(A_4 = A_2 +k) + p(A_4=A_1-k-1),
\end{equation}
and 
\begin{equation}
    \mathbb{Q}^2_k = p(A_3 = A_1 - k - 1 ) + p(A_3 = A_2 + k + 1) + p(A_4 = A_2 -k-1) + p(A_4 = A_1 + k),
\end{equation}
with
\begin{align}
    \alpha_k = \frac{1}{2d}[g(k) + (-1)^d \tan \left(\frac{\pi}{4d}\right)],~~~~~ \beta_k = \frac{1}{2d}[g(k+1/2) - (-1)^d \tan \left(\frac{\pi}{4d}\right)].
    \label{alphabeta}
\end{align}
Here, $g(k)=\cot [\pi(k+1/4)/d]$. Also, $p(A_i=A_j + k) := \sum_{m=0}^{d-1} p(a_i = m + k \, \, \text{mod} \, \, d, a_j = m|A_i, A_j)$, where  $p(a_i = m + k \, \, \text{mod} \, \,  d, a_j = m|A_i, A_j)$ denotes the joint probability of getting the outcome $a_i = (m + k \, \, \text{mod} \, \, d)$ when the measurement of $A_i$ is performed on the initially prepared state $\rho^{(\mathcal{P})}$ and the outcome $a_j=m$ when the measurement of $A_j$ is performed on the post measurement state of $A_i$. Following this notation, $\mathbb{P}^1_k$ and $\mathbb{Q}^1_k$ involve the probability distributions for the experimental runs in which $A_1$ or $A_2$ is measured at first on the initial preparation $\rho^{(\mathcal{P})}$ and then $A_3$ or $A_4$ is  measured on the post measurement state. Similarly, $\mathbb{P}^2_k$ and $\mathbb{Q}^2_k$ contain the probability distributions for the experimental runs in which $A_3$ or $A_4$ is measured at first on the initial preparation $\rho^{(\mathcal{P})}$ and then $A_1$ or $A_2$ is  measured on the post measurement state.

Now, following the calculations mentioned in the supplementary material of the Ref. \cite{SATWAP}, it can be shown that 
\begin{equation}
    \tau_d = d \, \widetilde{\tau}_d - 8S
    \label{real}
\end{equation}
with
\begin{equation}
    S= \frac{1}{2} \Bigg\lbrace 1 -  \cot \left[\frac{\pi}{d} \left( \floor*{\frac{d}{2}} + \frac{1}{4}\right)\right] \Bigg\rbrace .
    \label{sequation}
\end{equation}
Here, $\tau_d$ is the left hand sides of the temporal inequalities (\ref{temp}), i.e., 
\begin{align}
\tau_d = \sum_{k=1}^{d-1} &\Bigg[ a_k \left\langle A_1^{(k)} \, A_3^{(d-k)} \right\rangle + a_k^{*} \omega^k \left\langle A_1^{(k)}  \, A_4^{(d-k)} \right\rangle + a_k^{*} \left\langle A_2^{(k)} \, A_3^{(d-k)} \right\rangle  + a_k \left\langle A_2^{(k)} \, A_4^{(d-k)} \right\rangle  \nonumber \\
& + 
a_k \left\langle  A_3^{(d-k)} \, A_1^{(k)} \right\rangle + a_k^{*} \omega^k \left\langle A_4^{(d-k)} \, A_1^{(k)}  \right\rangle + a_k^{*} \left\langle A_3^{(d-k)} \,  A_2^{(k)} \right\rangle + a_k \left\langle A_4^{(d-k)} \, A_2^{(k)}  \right\rangle \Bigg]. \nonumber 
\end{align} 

One can see that $\widetilde{\tau}_d$  is actually a linear function of probability distributions $p(a_i,a_j|A_i,A_j)$ with real coefficients. Hence,  Eq.(\ref{real}) implies that the left hand side of the temporal inequality (\ref{temp}) always takes real values.\\

Let us now concentrate on the classical bound of the inequality \eqref{temp} proposed in the main text. Note that $\widetilde{\tau}_d$ can be expressed in the following alternative form,
\begin{align}
\widetilde{\tau}_d = \sum_{k=0}^{d-1} \alpha_k \Big[ &p(A_1 = A_3 +k) + p(A_2 = A_4 +k)    +  p(A_3 = A_1 +k) + p(A_4 = A_2 +k)  \nonumber \\
&+ p(A_2 = A_3 - k ) +   p(A_1 = A_4 - k -1) + p(A_3 = A_2 - k ) +   p(A_4 = A_1 - k - 1) \Big].
\label{taualt}
\end{align}
This can be achieved using the relation that $\alpha_k = - \beta_{d-k-1}$. Hence, the terms of the sum which are attached with  $\beta_k$ can be shifted to indices $k = \floor*{d/2}, \cdots, d - 1$ and are now associated with an $\alpha_k$ \cite{SATWAP}. For the cases where $d$ is odd,  the term
$k = \floor*{d/2}$ disappears,  as $\alpha_{\floor*{d/2}}=0$.

Next, we  use the notion of ``macrorealism" \cite{lgi} which is the conjunction of the following two assumptions, in order to derive the classical bound of the temporal inequality (\ref{temp}): \textit{(i) Realism:} At any instant, irrespective of any measurement, a system is definitely in any one of the available states such that all its observable properties have definite values. \textit{(ii) Noninvasive measurability:} It is possible, in principle, to determine which of the states the system is in, without affecting the state itself or the system's subsequent evolution. This notion of ``macrorealism" is one of the central concepts underpinning the classical world view.

The conjunction of the assumptions `Realism' and `Noninvasive measurability' implies that the probability of getting the outcomes $a_i$ and $a_j$, when the measurements of $A_i$ and $A_j$, respectively, are performed, does not depend on the order of the two measurements. Mathematically, it implies that $p(a_i,a_j|A_i,A_j) = p(a_j,a_i|A_j,A_i)$. Using this and from the definition of $p(A_i=A_j + k)$, we have
\begin{align}
    p(A_j=A_i + k) &= p(A_i=A_j - k)  \nonumber \\
    &= p(A_i=A_j + d - k) \, \, \, \, \, \, \forall i, j \in \{1,2,3,4 \} \, \, \text{and} \, \, \forall k \in \{0, 1, \cdots, d-1\}. \nonumber
\end{align}
Hence, we have
\begin{align}
    \sum_{k=0}^{d-1} \Big[p(A_i=A_j + k) + p(A_j=A_i + k) \Big] = \sum_{k=0}^{d-1} 2 p(A_i=A_j + k)  \, \, \, \forall i, j \in \{1,2,3,4 \}.
    \label{moddsum}
\end{align}
Using (\ref{moddsum}) and (\ref{taualt}), the expression of $\widetilde{\tau}_d$ for classical systems (denoted by $\widetilde{\tau}_{d_C}$) becomes 
\begin{align}
\widetilde{\tau}_{d_C} = \sum_{k=0}^{d-1} 2 \alpha_k \Big[ &p(A_1 = A_3 +k) + p(A_2 = A_4 +k) + p(A_2 = A_3 - k ) +   p(A_1 = A_4 - k - 1)  \Big].
\label{taualt2}
\end{align}


Now, `Realism' implies that we can assign definite value to each of the  observables which is revealed as the outcome of the measurement. Let the assigned value of $A_i$ be denoted by $v_i$ for all $i\in \{1,2,3,4\}$. Also, `Noninvasive measurability' implies that $v_i$ remains unaffected whether or not any measurement is performed prior to the measurement of $A_i$. Next, let us assign
one value $q$ such that $p(A_i = A_j + k) = \delta_{k,q}$, where $\delta_{k,q}$ is the Kronecker delta function. Here $q$ depends on $A_i$ and $A_j$ but not all pairs of $A_i$ and $A_j$ appearing 
in $\widetilde{\tau}_{d_C}$. In this way, we can define four variables $q_i$ $\in$ $\{0, 1 . . . , d - 1\}$ satisfying the following conditions \cite{SATWAP},
\begin{align}
    v_1 - v_3 &= q_1, \nonumber \\
    v_3 - v_2 &= q_2, \nonumber \\
    v_2 - v_4 &= q_3, \nonumber \\
    v_4 - v_1 &= q_4 + 1. 
\end{align}
Due to the chained character of these equations, we have 
\begin{equation}
 q_4 = -1 - \sum_{i=1}^{3} q_i  ~\text{mod}~ d. \nonumber 
\end{equation}

With these, $\widetilde{\tau}_{d_C}$ becomes
\begin{align}
    \widetilde{\tau}_{d_C} = 2 \left( \sum_{i=1}^{3} \alpha_{q_i} + \alpha_{-1- \sum_{i=1}^{3} q_i ~\text{mod}~ d} \right),
\end{align}
where $\alpha_k$ is defined in Eq.(\ref{alphabeta}). Therefore, the classical bound of $\widetilde{\tau}_d$ is given by,
\begin{equation}
   \widetilde{C}_d = 2 \max_{0\leq q_1, q_2, q_3 \leq d-1} \left( \sum_{i=1}^{3} \alpha_{q_i} + \alpha_{-1- \sum_{i=1}^{3} q_i ~\text{mod}~ d} \right).
\end{equation}
Next, using Theorem 1 in the supplementary material of the Ref. \cite{SATWAP}, 
we have
\begin{equation}
     \max_{0\leq q_1, q_2, q_3 \leq d-1} \left( \sum_{i=1}^{3} \alpha_{q_i} + \alpha_{-1- \sum_{i=1}^{3} q_i ~\text{mod}~ d} \right) = 3 \alpha_0 + \alpha_{d-1}
     \label{maxclass}
\end{equation}
Hence, using Eqs.(\ref{alphabeta}), (\ref{real}) and (\ref{sequation}), the classical bounds $C_d$ of the temporal inequalities (\ref{temp}) are obtained as follows
\begin{align}
     C_d &= d  \widetilde{C}_d - 8S \nonumber \\
    & = 2 d (3 \alpha_0 + \alpha_{d-1}) - 8S \nonumber \\
    & =  3 \cot \left(\frac{\pi}{4d} \right) - \cot \left( \frac{3\pi}{4d} \right) -4.
\end{align}

\section{Proof of Lemma \ref{lemma2}}\label{app3}

Here we present an example in the case of $d=2$ to show that if we are not restricted to any assumption on the state preparation, then  there exists two different projective measurements such that those two measurements are not connected unitarily although $\tau_d = 4(d-1)$ is achieved.

For the first strategy, the prepared state is of the form 
\[
\begin{pmatrix}
\cos\theta \\
e^{\mathbbm{i}\phi}\sin\theta \\
0\\
\end{pmatrix}
\]
for any $\theta, \phi$, and the measurements are of the form $A_i=A_i^+-A_i^-$ such that $A^+_i= | u_i \rangle\langle u_i|$ and $A_i^- = \mathbbm{1} - |u_i \rangle\langle u_i|$ wherein
\[ |u_1\rangle = \begin{pmatrix}
1\\
0 \\
0\\
\end{pmatrix},~~
|u_2\rangle = \begin{pmatrix}
\cos{\pi/4}\\
\sin{\pi/4}\\
0\\
\end{pmatrix},~~
|u_3\rangle = \begin{pmatrix}
\cos{\pi/8}\\
\sin{\pi/8}\\
0\\
\end{pmatrix},~~
|u_4\rangle = \begin{pmatrix}
~\cos{\pi/8}\\
-\sin{\pi/8}\\
0\\
\end{pmatrix}.
\] 
For the second strategy, the prepared state is 
\[
\begin{pmatrix}
0.427 \\
-0.512 - \mathbbm{i} 0.548\\
~0.067 + \mathbbm{i} 0.747\\
\end{pmatrix}.
\]
and the measurements are of the form $A_i=A_i^+-A_i^-$ such that $A^+_i= | v_i \rangle\langle v_i|$ and $A_i^- = \mathbbm{1} - |v_i \rangle\langle v_i|$ wherein
\[ |v_1\rangle = \begin{pmatrix}
1\\
0 \\
0\\
\end{pmatrix},~~
|v_2\rangle = \begin{pmatrix}
0.582 \\
-0.275 + \mathbbm{i} 0.308\\
-0.264 + \mathbbm{i} 0.317\\
\end{pmatrix},~~
|v_3\rangle = \begin{pmatrix}
\cos{\pi/4}\\
1/2~e^{\mathbbm{i} \pi/4} \\
1/2~e^{\mathbbm{i} \pi/4} \\
\end{pmatrix},~~
|v_4\rangle = \begin{pmatrix}
0.910 \\
-0.135 - \mathbbm{i} 0.384\\
-0.104 -\mathbbm{i} 0.393\\
\end{pmatrix}.
\] 

We can verify that $\tau_2 = 4$ is achieved by both  the strategies, where $\tau_2$ is defined in  \eqref{tid2}. 
If there exists a unitary, say $U$, that transforms one set of measurements to the another, then 
the following relations must hold
\be 
U |u_i\rangle = |v_i\rangle, \quad \forall i =1,2,3,4. 
\ee 
However, the fact that $|\langle u_1|u_3\rangle | \neq |\langle v_1|v_3\rangle |$ implies that such a unitary does not exist.  Note that the above example is in the three-dimensional binary outcome  projective measurement scenario. Other
such examples can easily be constructed in higher dimensions.

\section{Proof of Theorem \ref{thm2n}}\label{app4n}

Suppose the non-ideal observables $\widetilde{A}_1$, $\widetilde{A}_2$, $\widetilde{A}_3$, $\widetilde{A}_4$ satisfy the following,
\begin{align}
p(a_i | \widetilde{A}_i) -	p(a_i, a_i | \widetilde{A}_i, \widetilde{A}_i)   &= \eta_i^{(a_i)} \hspace{0.2cm} \text{with} \hspace{0.2cm} \eta_i^{(a_i)}>0  \hspace{0.4cm} \forall \, \, a_i \in \{0, \cdots, d-1\}, \, \, \forall \, \, i \in \{1, 2, 3, 4\}.
	\label{lemma1er}
\end{align}

Since condition \eqref{conjtprob} holds for all $|\psi\rangle \in \mathbb{C}^D$ (where $D$ is arbitrary), we have 
\begin{align}
p(a_i | \widetilde{A}_i) -	p(a_i, a_i | \widetilde{A}_i, \widetilde{A}_i) \geq 0 \hspace{0.3cm} \forall \, \,  |\psi\rangle \in \mathbb{C}^D.
	\label{lemma1e2r}
\end{align}
Hence, the condition (\ref{lemma1er}) together with (\ref{lemma1e2r})  under Assumptions \ref{ass1}-\ref{ass2}  implies that 
\begin{align}
0 \leq \langle \psi | \widetilde{K}_i^{a_i^{\dagger}}  \widetilde{K}_i^{a_i}   -	 \widetilde{K}_i^{a_i^{\dagger}} \widetilde{K}_i^{a_i^{\dagger}} \widetilde{K}_i^{a_i}  \widetilde{K}_i^{a_i}  | \psi \rangle < D \eta_i^{(a_i)} \hspace{0.3cm} \forall \, \,  |\psi\rangle \in \mathbb{C}^D,  \, \, \forall \, \, a_i \in \{0,  \cdots, d-1\}, \, \, \forall \, \, i \in \{1, 2, 3, 4\}.
\label{lemma1r3bb}
\end{align}
Hence, condition (\ref{lemma1r3bb}) implies the following,
\begin{align}
 0 \leq  \widetilde{K}_i^{a_i^{\dagger}}  \widetilde{K}_i^{a_i}   -	 \widetilde{K}_i^{a_i^{\dagger}} \widetilde{K}_i^{a_i^{\dagger}} \widetilde{K}_i^{a_i}  \widetilde{K}_i^{a_i}  < D \eta_i^{(a_i)} \mathbbm{1}  \hspace{0.3cm} \forall \, \, a_i \in \{0,  \cdots, d-1\}, \, \, \forall \, \, i \in \{1, 2, 3, 4\}.
	\label{lemma1r3}
\end{align}

Next, let us define the following positive Hermitian operator acting on $ \mathbb{C}^D$ for all $a_i \in \{0,  \cdots, d-1\}$ and for all $i \in \{1, 2, 3, 4\}$,
\begin{align}
	W^{a_i}_i = \widetilde{K}_i^{a_i^{\dagger}}  \widetilde{K}_i^{a_i}   -	 \widetilde{K}_i^{a_i^{\dagger}} \widetilde{K}_i^{a_i^{\dagger}} \widetilde{K}_i^{a_i}  \widetilde{K}_i^{a_i}.
	\label{lero4}
\end{align}
Also, consider that $\lambda_{1_i}^{a_i}$,  $\cdots$, $\lambda_{D-1_i}^{a_i}$ $\in \mathbb{R}$ are the eigenvalues of $W^{a_i}_i$. Hence, from (\ref{lemma1r3}), we can write
\begin{align}
0 \leq \lambda_{x_i}^{a_i} < D  \eta_i^{(a_i)} \hspace{0.3cm} \forall \, \, x \in \{0, \cdots, D-1\}, \, \,  \forall \, \, a_i \in \{0,  \cdots, d-1\}, \, \, \forall \, \, i \in \{1, 2, 3, 4\}.
\label{lero5}
\end{align}

Now, from the definition of Hilbert Schmidt norm, we have for all $a_i \in \{0,  \cdots, d-1\}$ and for all $i \in \{1, 2, 3, 4\}$
\begin{align}
\Big| \Big|    W^{a_i}_i  \Big| \Big|_{\text{HS}} = \sqrt{\sum_{x=0}^{D-1} \left( \lambda_{x_i}^{a_i} \right)^2} 
\label{lero6}
\end{align}

Now, using (\ref{lero4}-\ref{lero6}), we can write
\begin{align}
\Big| \Big|    \widetilde{K}_i^{a_i^{\dagger}}  \widetilde{K}_i^{a_i}   -	 \widetilde{K}_i^{a_i^{\dagger}} \widetilde{K}_i^{a_i^{\dagger}} \widetilde{K}_i^{a_i}  \widetilde{K}_i^{a_i}  \Big| \Big|_{\text{HS}} < \eta_i^{(a_i)} D^{\frac{3}{2}} \hspace{0.3cm} \forall \, \, a_i \in \{0,  \cdots, d-1\}, \, \, \forall \, \, i \in \{1, 2, 3, 4\}.
\label{lero7}
\end{align}

 Next, for the ideal  observables $A_i$ (with $i \in \{1,2,3,4\}$) satisfying (\ref{lemma1e}), we have that $K_i^{a_i^{\dagger}}  K_i^{a_i}  = \Big( K_i^{a_i^{\dagger}} \Big)^2 \Big( K_i^{a_i}  \Big)^2$ for all $a_i \in \{0, \cdots, d-1\}$. This follows from the fact that each of the ideal measurements satisfies the condition (\ref{lemma1e4}) mentioned in Appendix \ref{app1} for all $a_i \in \{0, \cdots, d-1\}$. Incorporating this relation into Eq.(\ref{lero7}), we have 
\begin{align}
	&\left| \left|  \rho^{(\mathcal{P})} \left\{ \widetilde{K}_i^{a_i^{\dagger}}  \widetilde{K}_i^{a_i}   -	 \Big( \widetilde{K}_i^{a_i^{\dagger}} \Big)^2 \Big( \widetilde{K}_i^{a_i}  \Big)^2  \right\} \rho^{(\mathcal{P})} -  \rho^{(\mathcal{P})} \left\{ K_i^{a_i^{\dagger}}  K_i^{a_i}   -	 \Big( K_i^{a_i^{\dagger}} \Big)^2 \Big( K_i^{a_i}  \Big)^2  \right\} \rho^{(\mathcal{P})} \right| \right|_{\text{HS}} <  \frac{\eta_i^{(a_i)}}{\sqrt{D}} < \eta_i^{(a_i)} \nonumber \\
	& \hspace{8.3cm} \forall \, \, a_i \in \{0,  \cdots, d-1\}, \, \, \forall \, \, i \in \{1, 2, 3, 4\}.
	\label{lero8}
\end{align}

\section{Proof of Theorem \ref{thm2}}\label{app4}

We first take  the condition (\ref{lemma1e}) to be satisfied by the each of the four non-ideal  observables $\widetilde{A}_1$, $\widetilde{A}_2$, $\widetilde{A}_3$, $\widetilde{A}_4$  under Assumptions \ref{ass1}-\ref{ass2}.  Hence, the measurement effects of each of these observables are mutually orthogonal projectors, which follows from Lemma \ref{lemma1}. Hence,  $\widetilde{A}_x^{(k)}$  is  unitary operator and $\widetilde{A}_x^{(k)} = \widetilde{A}_x^k$ for all $x \in \{1,2,3,4\}$ and for all $k \in \{0, \cdots, d-1\}$. Also it implies that the condition (\ref{UMU}) holds true for all the four non-ideal observables.  With these non-ideal unitary observables, we have  $\text{Tr}[\hat{\beta}_{\tau_d}\rho^{(\mathcal{P})}]= 4(d-1)-\epsilon$ (with $\epsilon > 0$), where $\hat{\beta}_{\tau_d}$ is the operator given by Eq.(\ref{temp3}) with $\widetilde{A}_1$, $\widetilde{A}_2$, $\widetilde{A}_3$, $\widetilde{A}_4$. Similar to \eqref{defBi}, we define
\begin{align}
&\widetilde{B}_1^{(k)} = a_k \widetilde{A}_3^{-k} + a_k^{*} \omega^k \widetilde{A}_4^{-k}, \nonumber \\
&\widetilde{B}_2^{(k)} = a_k^{*} \widetilde{A}_3^{-k} + a_k  \widetilde{A}_4^{-k} .
\end{align}
Let us also define the following,
\begin{equation}
\widetilde{P}_x^{(k)} = \mathbbm{1} - \widetilde{A}_x^k \, \widetilde{B}_x^{(k)} \hspace{0.3cm} \forall x \in \{1, 2\}, \, \, \forall k \in \{1, \cdots, d-1\}. 
\label{pxkrob}
\end{equation}
Now, following the calculation similar to that in Sec. \ref{sossection}, we get Eq.(\ref{sos5}) involving $\widetilde{P}_x^{(k)}$ and the operator $\hat{\beta}_{\tau_d}$ with $\widetilde{A}_1$, $\widetilde{A}_2$, $\widetilde{A}_3$, $\widetilde{A}_4$. From this, we have
\begin{align}
\text{Tr} &\Big[ \rho^{(\mathcal{P})} \sum_{k=1}^{d-1} \sum_{x=1}^{2} \Big[ \Big( \widetilde{P}_x^{(k)}\Big)^{\dagger} \, \Big(\widetilde{P}_x^{(k)}\Big) \Big] \Big] = \text{Tr} \Big[\rho^{(\mathcal{P})} ( 4 (d-1) \, \mathbbm{1} - \hat{\beta}_{\tau_d} )\Big]. \nonumber
\end{align}
Since $\rho^{(\mathcal{P})} = \mathbbm{1}/D$, we get,
\begin{align}
&\text{Tr} \Big[ \sum_{k=1}^{d-1} \sum_{x=1}^{2} \Big[ \Big( \widetilde{P}_x^{(k)}\Big)^{\dagger} \, \Big(\widetilde{P}_x^{(k)}\Big) \Big] \Big]= D \epsilon.
\end{align}
Being the sum of positive numbers, we get for the individual term,
\begin{align}
&\text{Tr} \Big[ \Big(\widetilde{P}_x^{(k)}\Big)^{\dagger} \, \Big(\widetilde{P}_x^{(k)}\Big) \Big]= f_{k,x}(\epsilon) \leq D \epsilon ~~~ \forall ~~ k,x .
\label{eqrob}
\end{align}
where $\sum_{k,x} f_{k,x}(\epsilon ) = D \epsilon$ and  $f_{k,x} (\epsilon) \geq 0$ for all $k,x$. Now, from
Eq.(\ref{eqrob}), we get
\begin{align}
\Big| \Big|    \widetilde{P}_x^{(k)}  \Big| \Big|_{\text{HS}} \leq   \sqrt{D \epsilon} ~~~ \forall ~~ k,x .
\end{align}
Hence, from Eq.(\ref{pxkrob}), we can write that
\begin{align}
\Big| \Big|   \mathbbm{1} - \widetilde{A}_x^k \, \widetilde{B}_x^{(k)}  \Big| \Big|_{\text{HS}} \leq    \sqrt{D \epsilon} ~~~ \forall ~~ k,x .
\label{condrob33n}
\end{align}
Now, for the ideal measurements, Eq.(\ref{p5n}) implies that $A_x^k \, B_x^{(k)} = \mathbbm{1}$ for all $k,x$. Thus, from the above relation we get
\be \label{robustnessiinmn}
 \Big| \Big|  A_x^k \, B_x^{(k)} -  \widetilde{A}_x^k \, \widetilde{B}_x^{(k)}  \Big| \Big|_{\text{HS}} \leq \sqrt{D \epsilon} \ \forall k,x .
 \ee 
Since $\rho^{(\mathcal{P})} = \mathbbm{1}/D$, we get the following from Eq.(\ref{robustnessiinmn}),
\begin{align} 
 \Big| \Big| \left( A_x^k \, B_x^{(k)} \right) \rho^{(\mathcal{P})} -  \left(\widetilde{A}_x^k \, \widetilde{B}_x^{(k)} \right) \rho^{(\mathcal{P})}  \Big| \Big|_{\text{HS}} &\leq \sqrt{\frac{\epsilon}{D}} \nonumber \\
 & < \sqrt{\epsilon} \, \, \, \, \forall k,x .
 \label{robustnessii}
 \end{align}
Putting $x=k=1$,  in Eq.(\ref{robustnessii}), we get Eqs.(\ref{robustness1}), and putting  $x=2,k=1$ in Eq.(\ref{robustnessii}), we get (\ref{robustness2}). These constitute a set of robustness arguments of our certification scheme.

We can derive another set of robustness arguments of our certification scheme as described below. Note that Eq.(\ref{eqrob}) also implies the following,
\begin{align}
\Big| \Big|    \Big(\widetilde{P}_x^{(k)} \Big)^{\dagger}  \Big| \Big|_{\text{HS}} &= \Big| \Big|   \mathbbm{1} -   \Big(\widetilde{B}_x^{(k)} \Big)^{\dagger} \, \Big(\widetilde{A}_x^k\Big)^{\dagger} \Big| \Big|_{\text{HS}} \nonumber \\
& \leq   \sqrt{D \epsilon} ~~~ \forall ~~ k,x ,
\label{condrob22n}
\end{align}
Next, we have for all $k,x$
\begin{align}
\Big| \Big|   \mathbbm{1} +   \Big(\widetilde{B}_x^{(k)} \Big)^{\dagger} \, \Big(\widetilde{A}_x^k\Big)^{\dagger} \Big| \Big|_{\text{HS}} 
&= \Big| \Big|   2\mathbbm{1} -\Big[\mathbbm{1} -   \Big(\widetilde{B}_x^{(k)} \Big)^{\dagger} \, \Big(\widetilde{A}_x^k\Big)^{\dagger} \Big] \Big| \Big|_{\text{HS}} \nonumber \\
 & \leq 2 \Big| \Big|   \mathbbm{1} \Big| \Big|_{\text{HS}} + \Big| \Big|  - \Big[ \mathbbm{1} -   \Big(\widetilde{B}_x^{(k)} \Big)^{\dagger} \, \Big(\widetilde{A}_x^k\Big)^{\dagger} \Big] \, \Big| \Big|_{\text{HS}}  \label{line2}\\
& \leq  \sqrt{D} \big(2 + \sqrt{ \epsilon} \big) \label{line4}
\end{align}
where we have used the triangle inequality for the Hilbert-Schmidt norm to get (\ref{line2}). We have also used (\ref{condrob22n}) and the fact that $\Big| \Big|   \mathbbm{1} \Big| \Big|_{\text{HS}} = \sqrt{D}$ to get (\ref{line4}).

In a similar way, it can be shown using (\ref{condrob33n}) that
\begin{align}
  \Big| \Big|   \mathbbm{1} + \widetilde{A}_x^k \, \widetilde{B}_x^{(k)}  \Big| \Big|_{\text{HS}} \leq  \sqrt{D} \big(2 + \sqrt{ \epsilon} \big) ~~~ \forall ~~ k,x .  
  \label{robsep4}
\end{align}
Now, we can obtain the following relation for all $k,x$,
\begin{align}
 &\Big| \Big| \mathbbm{1} - \Big(\widetilde{B}_x^{(k)} \Big)^{\dagger}  \Big(\widetilde{B}_x^{(k)} \Big) \Big| \Big|_{\text{HS}} \nonumber \\
 &=  \Big| \Big| \mathbbm{1} - \Big(\widetilde{B}_x^{(k)} \Big)^{\dagger}  \, \Big(\widetilde{A}_x^k\Big)^{\dagger} 
 \, \Big(\widetilde{A}_x^k\Big) \, \Big(\widetilde{B}_x^{(k)} \Big) \Big| \Big|_{\text{HS}} \label{2line2} \\ 
 & =  \Bigg| \Bigg| \frac{1}{2} \Bigg[\Bigg( \mathbbm{1} +   \Big(\widetilde{B}_x^{(k)} \Big)^{\dagger} \, \Big(\widetilde{A}_x^k\Big)^{\dagger} \Bigg) \Bigg( \mathbbm{1} - \widetilde{A}_x^k \, \widetilde{B}_x^{(k)} \Bigg) + \Bigg( \mathbbm{1} -   \Big(\widetilde{B}_x^{(k)} \Big)^{\dagger} \, \Big(\widetilde{A}_x^k\Big)^{\dagger} \Bigg) \Bigg( \mathbbm{1} + \widetilde{A}_x^k \, \widetilde{B}_x^{(k)} \Bigg) \Bigg] \Bigg| \Bigg|_{\text{HS}} \label{2line3}\\
 & \leq  \frac{1}{2}  \Bigg| \Bigg| \Bigg( \mathbbm{1} +   \Big(\widetilde{B}_x^{(k)} \Big)^{\dagger} \, \Big(\widetilde{A}_x^k\Big)^{\dagger} \Bigg) \Bigg( \mathbbm{1} - \widetilde{A}_x^k \, \widetilde{B}_x^{(k)} \Bigg) \Bigg| \Bigg|_{\text{HS}} + \frac{1}{2}\Bigg| \Bigg| \Bigg( \mathbbm{1} -   \Big(\widetilde{B}_x^{(k)} \Big)^{\dagger} \, \Big(\widetilde{A}_x^k\Big)^{\dagger} \Bigg) \Bigg( \mathbbm{1} + \widetilde{A}_x^k \, \widetilde{B}_x^{(k)} \Bigg)  \Bigg| \Bigg|_{\text{HS}} \label{2line4}\\
 & \leq  \frac{1}{2}  \Bigg| \Bigg| \mathbbm{1} +   \Big(\widetilde{B}_x^{(k)} \Big)^{\dagger} \, \Big(\widetilde{A}_x^k\Big)^{\dagger} \Bigg| \Bigg|_{\text{HS}} \, \Bigg| \Bigg| \mathbbm{1} - \widetilde{A}_x^k \, \widetilde{B}_x^{(k)}  \Bigg| \Bigg|_{\text{HS}} + \frac{1}{2} \Bigg| \Bigg|  \mathbbm{1} -   \Big(\widetilde{B}_x^{(k)} \Big)^{\dagger} \, \Big(\widetilde{A}_x^k\Big)^{\dagger} \Bigg| \Bigg|_{\text{HS}} \, \Bigg| \Bigg| \mathbbm{1} + \widetilde{A}_x^k \, \widetilde{B}_x^{(k)}   \Bigg| \Bigg|_{\text{HS}} \label{2line5}\\
 & \leq D \sqrt{\epsilon} \big(2 + \sqrt{\epsilon} \big) \label{2line6}
\end{align} 
where we have used $\Big(\widetilde{A}_x^k\Big)^{\dagger} \Big(\widetilde{A}_x^k\Big) = \mathbbm{1}$ (since,  $\widetilde{A}_x^k$ is a unitary operator). Next, the triangle inequality for the Hilbert-Schmidt norm is used to get (\ref{2line4}). We obtain (\ref{2line5}) as the Hilbert-Schmidt norm is a submultiplicative norm (i.e., for all $n \times n$ square matrices $A$ and $B$, $|| A B ||_{\text{HS}} \leq || A ||_{\text{HS}} \, \, || B ||_{\text{HS}}$). Finally, we have used  conditions (\ref{condrob33n}), (\ref{condrob22n}), (\ref{line4}), (\ref{robsep4}) to get (\ref{2line6}) from (\ref{2line5}).

Now,  for ideal measurements, Eq.(\ref{p9nnnn}) implies that $\Big(B_x^{(k)}\Big)^{\dagger}  \Big(B_x^{(k)} \Big) =  \mathbbm{1} $ for all $x \in \{1,2\}$ and for all $k \in \{1, \cdots, d-1\}$. Hence, from (\ref{2line6}), we get 
\begin{align}
\Big| \Big| \Big(B_x^{(k)}\Big)^{\dagger}  \Big(B_x^{(k)} \Big) - \Big(\widetilde{B}_x^{(k)} \Big)^{\dagger}  \Big(\widetilde{B}_x^{(k)} \Big) \Big| \Big|_{\text{HS}}  \leq D \sqrt{\epsilon} \big(2 + \sqrt{\epsilon} \big). 
\label{robustness2mainnn}
\end{align}
Next, considering $\rho^{(\mathcal{P})} = \mathbbm{1}/D$, we get the following from the above inequality,
\begin{align}
\Bigg| \Bigg| \left[\Big(B_x^{(k)}\Big)^{\dagger}  \Big(B_x^{(k)} \Big) \right] \rho^{(\mathcal{P})} - \left[\Big(\widetilde{B}_x^{(k)} \Big)^{\dagger}  \Big(\widetilde{B}_x^{(k)} \Big) \right] \rho^{(\mathcal{P})} \Big| \Big|_{\text{HS}}  \leq  \sqrt{\epsilon} \big(2 + \sqrt{\epsilon} \big). 
\label{robustness2main}
\end{align} 
Now, putting $x=1$ and $k=1$, we get \eqref{robustness1new} after some simplification.
Eq.(\ref{robustness1new}) presents another robustness argument of our certification scheme for  $A_3$ and $A_4$.

A similar procedure can be followed in order to establish robustness arguments for $A_1$ and $A_2$. For this, consider
\begin{equation}
    C_1^{(k)} = a_k A_1^{-k} + a_k^* A_2^{-k} ~~ \text {and} ~~
    C_2^{(k)} = \omega_k a_k^* A_1^{-k} + a_k A_2^{-k} .\nonumber
\end{equation}
Then, one can have another sum-of-squares decomposition of the inequality \eqref{temp} as
\begin{align}
&\sum_{k=1}^{d-1} \sum_{x=1}^{2} \Big[ \Big(Q_x^{(k)}\Big)^{\dagger} \, \Big(Q_x^{(k)}\Big) \Big] = 4 (d-1) \, \mathbbm{1} - \hat{\beta}_{\tau_d} \nonumber \\
&\hspace{1.5cm} \text{with} \hspace{0.15cm} Q_x^{(k)} = \mathbbm{1} - A_{x+2}^k \,  C_x^{(k)} \hspace{0.4cm} \forall x \in \{1, 2\}, \, \forall k \in \{1, \cdots, d-1\}.
\end{align}
Through a similar analysis as discussed above, we get the following relation for all $k,x$,
\begin{align}
\Bigg| \Bigg| \left[\Big(C_x^{(k)}\Big)^{\dagger}  \Big(C_x^{(k)} \Big) \right] \rho^{(\mathcal{P})} - \left[\Big(\widetilde{C}_x^{(k)} \Big)^{\dagger}  \Big(\widetilde{C}_x^{(k)} \Big) \right] \rho^{(\mathcal{P})} \Big| \Big|_{\text{HS}}  \leq  \sqrt{\epsilon} \big(2 + \sqrt{\epsilon} \big). 
 \label{robsepn2}
\end{align} 
Now, putting $x=1$ and $k=1$ in (\ref{robsepn2}), we get the robustness expression \eqref{robustness2new} involving $A_1$ and $A_2$.

\end{document}